\documentclass[letterpaper,10pt,twoside]{article}
\usepackage[utf8]{inputenc}
\usepackage[english]{babel}
\usepackage[T1]{fontenc}
\usepackage{lmodern}
\usepackage[mathscr]{eucal}
\usepackage{amssymb,amsmath,amsthm,mathtools}
\usepackage[knot,matrix]{xy}
\newcommand{\gf}{\varphi} 
\newcommand{\geps}{\varepsilon}

\theoremstyle{definition}
\newtheorem{definition}{Definition}[section]

\theoremstyle{plain}
\newtheorem{lemma}[definition]{Lemma}
\newtheorem{proposition}[definition]{Proposition}
\newtheorem{theorem}[definition]{Theorem}

\begin{document}

\title{\textbf{A Poisson Algebra for Abelian Yang-Mills Fields on Riemannian Manifolds with Boundary}}
\author{%
  Homero G. Díaz-Marín\\
 Facultad de Ciencias F\'\i sico-Matem\'aticas\\
 Universidad Michoacana de San Nicol\'as de Hidalgo\\
 Ciudad Universitaria, C.P.~58060,  Morelia, Michoac\'an, M\'exico\\
 {\it E-mail address}: {\tt hdiaz@umich.mx}
}

\maketitle

\begin{abstract}
We define a family of observables for abelian Yang-Mills fields associated to compact regions $U\subseteq M$ with smooth boundary in Riemannian manifolds. Each observable is parametrized by a first variation of solutions and arises as the integration of gauge invariant conserved current along admissible hypersurfaces contained in the region. The Poisson bracket uses the integration of a canonical presymplectic current.\\[0.2cm]
 
 
\end{abstract}

\label{first}

[]\section{Introduction}\label{hel}

In Classical Covariant Field Theory two desirable conditions are required for a family of observables: In one side we require this function to separate solutions of the Euler-Lagrange equations. On~the other hand, we need the Jacobi identity in order to have a Lie (Poisson) bracket. It is a known problem to characterize those theories accomplishing these two requirements, as~pointed out {in}~\cite{Sternberg-G, Kijowski} and others. There are two main difficulties. On~one hand, under~locality assumptions, Jacobi identity is well established but generically there are few observables associated with conservation laws given by Noether's First Theorem, see for instance~\cite{Deligne-Freed}. On~the other hand, extending to non-locality of variations of solutions, we may provide enough observables, see for {instance}~\cite{Lunev, Anco-Bluman}, nevertheless the Jacobi identity does not necessarily hold, see~\cite{Olver:Ghosts}.
For linear theories there are no such difficulties, and vector fields in the space of solutions can be modeled as in Theorem~\ref{tma:darboux}, {see also}~\cite{Zapata(2019)}. For~instance, in Lorentzian globally hyperbolic spacetimes, Maxwell equations~\cite{AB} exhibit a family of observables, related to the Aharomov-Bohm effect, and~a Poisson bracket constructed with Peierls method for local variables. We provide a similar set of  observables for the abelian Yang-Mills (YM) fields on Riemannian manifolds. This could be mentioned as the novelty introduced in this work, although~our aim is to prepare the scenario for non-abelian (non-linear) YM fields. \mbox{We adopt the} Lagrangian approach of the variational bicomplex formalism, see~\cite{Zuckerman, Vitagliano, Reyes} rather than the Hamiltonian multysimplectic formalism approach to describe non abelian YM fields, see~\cite{Helein:Covariant, Helein:YM}.

We consider regions $U$ with smooth boundary $\partial U$ both contained in a $n$-dimensional Riemannian manifold, usually $n=4$. Here we avoid the complications of corners in $\partial U$ which will be treated elsewhere. For~a principal bundle we take solutions of the Yang-Mills (YM) equations for the abelian $U(1)$ structure group. We are interested in defining a family of observables for YM solutions in $U$, $\eta\in{\cal A}_U$, of~the integral form
\[
	f_\Sigma(\eta)=\int_\Sigma {\sf j}\eta^*F
\]
defined for a $3$-dimensional compact Riemannian {\em admissible} smooth hypersurface $i_\Sigma:\xymatrix{\Sigma
\ar@{^{(}->}[r]& U}$ with volume form $\nu_\Sigma$, where admissibility means $\partial\Sigma \subseteq \partial U$, see~\cite{Forger}.  {\em Observable currents}, are horizontal $(n-1)$-forms, $F\in {\sf \Omega}^{n-1,0}({\sf J}Y\vert_U)$, in~the $\infty-$ jet bundle ${\sf J} Y$ associated to sections of the affine bundle $Y\rightarrow M$ of connections. The~local invariance condition is then assumed by imposing ${\sf d_h}F\vert_U=0$, when restricted to the locus of the YM equations ${\cal E}_L$. ${\sf d_h}$ is the horizontal differential, see the notation of the variational bicomplex formalism in Appendix \ref{VarBiC}. We adapt helicity for hypersurfaces embedded properly in general compact regions $U$, rather than considering cylinder regions with space-like slices, $\Sigma\times[t_1,t_2]$, this is related to the General Boundary Formalism for field theories, see~\cite{DM-O} and references~therein.

The idea is to define the relative helicity from hydrodynamics properly adapted to YM fields as a local observable. In~order to motivate this definition we recall the notion of helicity from magneto-hydrodynamics. For~a divergence (non-autonomous) free vector field, $\xi=\xi(t)\in \frak{X}_{\partial \Sigma}(\Sigma) $ in a three-dimensional Riemannian manifold $\Sigma$ tangent to the boundary $\partial \Sigma$, helicity is defined as
\begin{equation}\label{eqn:hel0}
	\int_\Sigma \overline{g}(v,\xi)\nu_\Sigma
\end{equation}
where one considers the vector field $v=v(t)$, as~a potential in $\Sigma$. Helicity of $\xi$ measures globally the degree of self-linking of its flow. Helicity remains an invariant for every $\nu_\Sigma$-preserving diffeomorphism of $\Sigma$ that carries the boundary $\partial \Sigma$ into itself, where $\nu_\Sigma$ is given by the volume form on $\Sigma$. The~situation can be dually described in terms of $1$-forms. If~$\alpha=\overline{g}(v,\cdot)$ where $\overline{g}$ is the Riemannian metric on $\Sigma$, then under the additional topological condition, $H^2_{\rm dR}(\Sigma)=0$, there exists a potential $\alpha\in {\sf \Omega}^1(\Sigma)$ such that $d\alpha=\iota_\xi\nu_\Sigma$. Here helicity reads as
\begin{equation}\label{eqn:hel-1}
	\int_\Sigma \alpha\wedge d\alpha
\end{equation}

It does depend just on the vorticity $d\alpha$ not on the $1$-form $\alpha$ or the vector field $v$. If we adopt $v$ a divergence-free vector field or $d^{\star_\Sigma}({\alpha})=0, $ respectively, then the property of {\em isovorticity} holds for $v(t)$ for the magnetic potential, as~well as for any solution of the Euler equation of hydrodynamics. This means that $\xi(t_2)$ can be constructed as the image of $\xi(t_1)$ under a diffeomorphism and if we consider a space-time domain $\Sigma\times[t_1,t_2]$, then helicity does not depend on the parameter $t$ of the non-autonomous flow. To~review this concepts see for instance~\cite{Arnold-Khesin, Khesin}.

Under the assumption of simple connectednes of $\Sigma$, then the Lie algebra of divergence-free vector fields, have a bilinear form, {\em relative helicity}, defined as
\[
	[\alpha,\beta]_\Sigma=\int_\Sigma d\alpha \wedge \beta
\]

Notice that helicity is $[\alpha,\alpha]_\Sigma$ and also that $[\cdot,\cdot]_\Sigma$ is a symmetric bilinear form under the assumption of closedness for $\Sigma$.

Considering YM solutions $\eta=\eta_0+\gf\in{\cal A}_U$, where $\eta_0\in{\cal A}_U$ is a fixed connection and $\gf=\eta-\eta_0$ is a 1-form in $M$, we would like to define the {\em field strength helicity} as in (\ref{eqn:hel-1}). Choose a tubular neighborhood $\Sigma_\geps\subseteq U$ of $\Sigma^0:=i_\Sigma(\Sigma,\tau)$ with exponential coordinates $X_\Sigma: \Sigma\times [-\geps,\geps] \rightarrow \Sigma_\geps$, with~embedding $i_\Sigma=X_\Sigma(\cdot,0)$. We take $\tilde{\gf}=\gf+df \vert_{\Sigma_\geps}$ an {\em axial gauge fixing}, that is a $1$-form such that in $\Sigma^0$ has no normal component. In~addition, we may suppose that $\psi^{\Sigma^0}_\eta=i_{\Sigma}^*\tilde{\gf}$, as~well as $\frac{d}{d\tau}\vert_{\tau=0} \psi^{\Sigma^\tau}_\eta$ are divergence-free. See Appendix on the geometry of abelian YM fields in~\cite{DM-O}.
 
Then the helicity for abelian YM fields could be defined as
\[
	[\tilde{\gf},\tilde{\gf}]_\Sigma=\int_\Sigma \psi^{\Sigma^0}_\eta \wedge \star_\Sigma 
	 \left.\frac{d\psi^{\Sigma^\tau}_\eta}{d\tau}\right\vert_{\tau=0},
\]
where $\star_\Sigma$ is the Hodge operator associated to the induced Riemannian metric $\overline{g}$ on $\Sigma$. Hence we could define helicity as in (\ref{eqn:hel0}) for the vector fields $v,\xi$ defined as
$
	\overline{g}(v,\cdot)=\psi^{\Sigma^0}_\eta,
	\quad
 	\overline{g}(\xi,\cdot)=\left.
 	\frac{d\psi^{\Sigma^\tau}_\eta}{d\tau}
	\right\vert_{\tau=0}
	.
$

Nevertheless, this notion of helicity would depend on the gauge fixing choice, therefore {\em cannot be generalized} as a gauge invariant observable. Moreover, we do not get a local ${\sf d_h}$-closedness condition for an observable current: if $U'\subseteq U$ is an open region such that $\partial U'=\Sigma-\Sigma'$, then 
\[
	[\tilde{\gf},\tilde{\gf}]_\Sigma
	=
	[\tilde{\gf},\tilde{\gf}]_{\Sigma'}+ \int_{U'} L({\sf j}\eta)
\]
where $L$ is the Lagrangian density. We will rather try to define the {\em relative helicity} of YM fields. Take $\eta'=\eta_0+\gf'\in {\cal A}_U$ any other solution. Take a first variation of solutions $\gf$, let us define
\[
	[\gf,{\gf}']_\Sigma
	=
	\int_\Sigma
		i_\Sigma^*({\gf}'\wedge \star d{\gf})
		.
\]

Then for gauge translations $\eta'+df$ we would have $[\gf,\gf']_\Sigma=[\gf,{\gf}'+df]_\Sigma$. Moreover, if~$U'\subseteq U$ is an open region such that $\partial U'=\Sigma-\Sigma'$, then 
$
	[\gf,{\gf}']_\Sigma
	=
	[\gf,{\gf}']_{\Sigma'}
	.
$
Thus for every couple $\eta,\gf$ where $\eta\in{\cal A}_U$ and $\gf$ is a first variations of solutions, we consider the {\em antisymmetric component of the relative helicity} or simply $\gf$-{\em helicity},
\begin{equation}\label{eqn:hel2}
	h^\gf_\Sigma(\eta')
	:=
	\frac{1}{2}\left(
		[\gf,{\gf}']_\Sigma
	-
	[\gf',{\gf}]_\Sigma
	\right).
\end{equation}

In Section~\ref{sec:lie} we formalize this construction in the language of the variational bicomplex, see Appendix~\ref{VarBiC}.


\section{Variational Bicomplex Formalism for Abelian YM~Fields}\label{sec:bicomplex}


Along this section we adopt the terminology and notation of the variational bicomplex formalism, for the readers convenience we give a brief presentation and references for this in Appendix~\ref{VarBiC}. Let ${\cal P}\rightarrow M$ be a principal bundle on a Riemannian manifold $(M,g)$ with structure group $G=U(1)$ and $U\subseteq M$ a region with smooth boundary. Let $\pi:Y\rightarrow M$ with $Y={\sf J}^1{\cal P}/G$ be the affine bundle whose sections ${\Gamma}(Y)$ are the $G$-covariant connections on ${\cal P}$.

For abelian YM, the~{\em Lagrangian density} $L={\rm L}\,\nu \in\mathsf{ \Omega}^{n,0}({\sf J}Y)$ is defined by the Lagrangian
\[
	{\rm L}
	=
	\frac{1}{4}
	\sum_{i,j=1}^n\varpi^{ij} \left(A^j_i - A^i_j\right)^2
\]
where this expresion corresponds to local coordinates $(x_1,\dots,x_n;A^1,\dots,A^n;A^i_j)$ in ${\sf J}^1Y$, $\nu=dx^1\wedge\cdots \wedge d x^n$ is a fixed volume form in the base and $\varpi^{ij} =\sqrt{|\det\, g|}\,g^{ii}g^{jj}$, with $g_{ij}$ the Riemannian metric in $U$.

Then $E(L)=\sum_{i=1}^nE_i({\rm L})\vartheta^i\wedge \nu$ denote the Euler-Lagrange equations, where $\vartheta^i=\mathsf{d_v}A^i$ stands for the basis for the vertical $1-$forms in ${\sf J}^1Y$. Thus {\em YM equations} have {\em locus} which is the prolongation $ {\cal E}_L\subseteq {\sf J}Y$ of $\{E(L)=0\}\subseteq {\sf J}^2Y$. In~the local coordinate chart,
\[
	E_i({\rm L})=
		\sum_{j=1}^n\frac{\sf d}{{\sf d}{x^j}}
		\left(
			\varpi^{ij}\left( A^j_i - A^i_j\right)
		\right)=0,
		\quad \forall i=1,\dots,n.
\]

The {\em space of solutions} over $U$ is
\begin{equation}\label{eqn:A_U}
	{\cal A}_U
	=
	\left\{
		\eta\in \Gamma(Y\vert_U)\,:\,
		{\sf j}\eta(x)\in {\cal E}_L
	\right\}
\end{equation}

Thus solutions $\eta$ satisfy ${\sf j}\eta^*E(L)=0$.

The {\em linearized equations} for any (local) {\em evolutionary vector field,} $V\in \mathfrak{Ev}({\sf J} Y),$ are
\begin{equation}\label{eqn:linearized-EL}
	{\sf I}({\sf d_v}\iota_{{\sf j}V}E(L))\vert_{{\cal E}_L}=0
\end{equation}
where ${\sf I}:{\sf \Omega}^{n,k}({\sf J}Y)\rightarrow {\sf \Omega}^{n,k}({\sf J}Y)$ is the {\em integration by parts operator}, see its definition {in}~\cite{Anderson}. In~local coordinates this linearized equation reads as
\[
		\left.\sum_{j=1}^n
		\frac{\sf d}{{\sf d}{x^j}}
		\left(
			\varpi^{ij}\left( \frac{\sf d}{{\sf d}{x^i}}V^j -
			\frac{\sf d}{{\sf d}{x^j}}	V^i\right)
		\right)\right\vert_{{\cal E}_L}=0,
		\quad \forall i=1,\dots,n,
		\, V=\sum_{i=1}^n V^i\frac{\partial }{\partial A^i}.
\]

Let ${\frak F}_U\subseteq \mathfrak{Ev}({\sf J} Y)$ be the Lie subalgebra of those evolutionary vector fields satisfying the linearized Euler-Lagrange equations. The~Lie algebra ${\frak F}_U$ will turn out to be our model for {\em variations of YM solutions}. For~example, the radial evolutionary vector field $R=\sum_aA^a\frac{\partial}{\partial A^a}$ whose prolongation is
\begin{equation}\label{eqn:R}
	{\sf j}R=\sum_aA^a\frac{\partial}{\partial A^a}
	+\sum_iA^a_i\frac{\partial}{\partial A^a_i}+\dots
\end{equation}
is a symmetry of the YM PDE, i.e. $R\in{\frak F}_U$. This is a general constructions of symmetries for linear PDEs, see~\cite{Anco-Bluman}.

The {\em presymplectic current}
\[
	\Omega_L = \sum_{i,j=1}^n
		\varpi^{ij}(\vartheta^i_j- \vartheta^j_i)\wedge\vartheta^j\wedge \nu^i
\]
with $dx^i\wedge \nu^i=\nu$, has the property stated in the following general~Lemma.

\begin{lemma}[Multysimplectic formula]
For every $V,W\in{\frak F}_U$ we have 
\[
	{\sf d_h}\left(
		\iota_{{\sf j}W}		\iota_{{\sf j}V} \Omega_L\right) \vert_{{\cal E}_L}
	= 0.
\]

\end{lemma}

\begin{definition}[Gauge]   \label{dfn:hat-gauge}
\text{~~} 
\begin{enumerate}
\item Those first variations of solutions $V\in {\frak F}_U$ satisfying
\begin{equation}\label{eqn:LH}
	\iota_{{\sf j}W}\mathscr{L}_{{\sf j}V}\Omega_L\vert_{{\cal E}_L}
	=
		\iota_{{\sf j}W}{\sf d_v}( \iota_{{\sf j}V}\Omega_L)\vert_{{\cal E}_L}
	= \iota_{{\sf j}W}{\sf d_h}\sigma^V,
	\qquad
	\forall W\in{\frak F}_U,
\end{equation}
define the Lie subalgebra  of {\em locally Hamiltonian first variations} as $\hat{\frak F}^{\rm LH}_U\subseteq {\frak F}_U$.

\item We define the Lie algebra $\hat{\frak G}_U$ of {\em gauge first variations}  as those $X\in {\frak F}_U$ satisfying locally the presymplectic degeneracy condition, i.e.,
\begin{equation}\label{eqn:gauge-d_h}
	\iota_{{\sf j} W}\left.\left(
		\iota_{{\sf j} X}\Omega_L
	\right)\right\vert_{{\cal E}_L}
	=
	\iota_{{\sf j} W}
		{\sf d_h} \rho^X,
		\qquad
	\forall W\in{\frak F}_U.
\end{equation}

\end{enumerate}
\end{definition}

For instance, the~radial vector $R\in{\frak F}_U$ defined in (\ref{eqn:R}) is not locally hamiltonian, since it satisfies the Liouville condition $\iota_{{\sf j}R}{\sf d_v}\Omega_L=\Omega_L,$ rather than condition (\ref{eqn:LH}).

In the second part of Definition \ref{dfn:hat-gauge} we may also have adopted $X\in \mathfrak{Ev}({\sf J}Y)$ instead of $X\in {\frak F}_U$ and
\[
	\iota_{{\sf j} W}\left.\left(
		\iota_{{\sf j} X}\Omega_L
	\right)\right\vert_{{\cal E}_L}
	=
	\iota_{{\sf j} W}
		{\sf d_h} \rho^X,
		\qquad
	\forall W\in \mathfrak{Ev}({\sf J}Y\vert_U)
\]
as is stated in the following~assertion.

\begin{proposition}\label{prop:LH}
Suppose that $V\in\mathfrak{Ev}({\sf J}Y\vert_U)$ satisfies
\[
	\iota_{{\sf j}W}{\sf d_h}(\iota_{{\sf j}V}\Omega_L)\vert_{{\cal E}_L}=0
\]
for every variation of solutions $W\in {\frak F}_U$. Then $V\in{\frak F}_U$.
\end{proposition}

Notice that the locally Hamiltonian condition is stronger than the property exhibited in \mbox{Proposition \ref{prop:LH}} for every variation of solutions. Thus $\hat{\frak G}_U\subseteq \hat{\frak F}^{\rm LH}_U$.

\begin{lemma}
$\hat{\frak G}_U\subseteq \hat{\frak F}^{\rm LH}_U$ is a Lie ideal.
\end{lemma}

\begin{proof}
If $X,X'\in  \hat{\frak G}_U$ then
\begin{equation}\label{eqn:basic}
	\iota_{{\sf j} [X,X']}\Omega_L=
	{\sf d_v}\iota_{{\sf j}X}\iota_{{\sf j}X'}\Omega_L
	\pm
	\iota_{{\sf j}X}{\sf d_v}\iota_{{\sf j}X'}\Omega_L
	\pm
	\iota_{{\sf j}X'}{\sf d_v}\iota_{{\sf j}X}\Omega_L
\end{equation}
which by hypothesis and by anticommutativity of ${\sf d_vd_h}=-{\sf d_hd_v}$ is ${\sf dh}$-exact, hence $[X,X']\in \hat{\frak G}_U$ and therefore $\hat{\frak G}_U\subseteq {\frak F}_U$ is a Lie subalgebra. To~see that $\hat{\frak G}_U\subseteq \hat{\frak F}^{\rm LH}_U$, apply vertical derivation \mbox{to (\ref{eqn:gauge-d_h}).}

Take $V\in\hat{\frak F}^{\rm LH}_U$, then $[V,X]$ apply vertical derivation to Equation~(\ref{eqn:basic}) with $V=X'$ and the condition of ${\sf d_h}$-exactness for $\iota_{{\sf j}X} {\sf d_v}\iota_{{\sf j}V}\Omega_L$ implies the ${\sf d_h}$-exactness of $\iota_{{\sf j}[X,V]}\Omega_L$ holds. Therefore $[V,X]\in\hat{\frak G}_U$.
\end{proof}

Form Proposition~\ref{prop:LH} it follows also the following~assertion.

\begin{lemma}
$\hat{\frak G}_U\subseteq \mathfrak{Ev}({\sf J}Y\vert_U)$ is a Lie ideal, hence $\hat{\frak F}^{\rm LH}_U/\hat{\frak G}_U\subseteq \mathfrak{Ev}({\sf J}Y\vert_U)/\hat{\frak G}_U.$
\end{lemma}

\begin{lemma}\label{lma:X=drho}
If for every $W\in{\frak F}_U$
\[
	\iota_{{\sf j}W}(\iota_{{\sf j}X}\Omega_L)\vert_{{\cal E}_L}=\iota_{{\sf j}W}{\sf d_h}\rho^X\vert_{{\cal E}_L},\quad
	\rho^X=\sum_{i,j=1}^n\rho^{ij}\vartheta^j_i\wedge \nu^{ji}.
	\]
holds, then in local coordinates $\varpi^{ij}X^
j=\frac{\sf d}{{\sf d}x^j}\rho^{ij}$ holds in ${\cal E}_L$ for each $i,j=1,\dots, n,$ 
where $dx^j\wedge\nu^{ji}=\nu^i$.

\end{lemma}

\begin{definition}[Gauge with boundary condition]
\text{~~}
\begin{enumerate}
\item The Lie subalgebra 
\[
	{\frak F}^{\rm LH}_U\subseteq
	\hat{\frak F}^{\rm LH}_U
\]
of locally Hamiltonian first variations {\em with null boundary conditions}, consists of those $V\in \hat{\frak F}^{\rm LH}_U$ satisfying (\ref{eqn:LH}) and
\[
\sigma^V\vert_{\partial U}={\sf d_h}\lambda\vert_{\partial U}.
\]
when evaluated in ${\cal E}_L,{\frak F}_U$. In~particular $\left.\mathscr{L}_{ {\sf j} V}\Omega_L\right\vert_{\partial U}=0$.

\item The Lie ideal of {\em gauge variations with null boundary conditions}
\[
	{\frak G}_U\subseteq {\frak F}^{\rm LH}_U
\]consists of those $X\in\hat{\frak G}_U $ such that (\ref{eqn:gauge-d_h}) holds together with
\[
	{\sf j}X\vert_{\partial U}=0.
\]
Which means that there is no gauge action in the boundary.
\end{enumerate}
\end{definition}

The following assertions are used in the~definition.

\begin{lemma}
The following inclusions  are Lie ideal inclusions into Lie algebras: 
\[
	{\frak G}_U\subseteq {\frak F}^{\rm LH}_U,
	\,
	{\frak G}_U\subseteq \hat{\frak F}^{\rm LH}_U,
	\,
	{\frak G}_U\subseteq \hat{\frak G}_U,
	\,
	{\frak F}^{\rm LH}_U\cap \hat{\frak G}_U\subseteq {\frak F}_U^{\rm LH}.
\]
\end{lemma}

\begin{proof}
$X,X'\in{\frak G}_U$ imply that ${\sf j}[X,X']\vert_{\partial U}=0$ hence ${\frak G}_U$ is indeed a Lie algebra. To~see that it is an ideal in $\hat{\frak F}_U^{\rm LH}$ we just consider the fact that ${\sf j}[X,V]\vert_{\partial U}=0$ for every $V\in{\frak F}_U^{\rm LH}$.

To see that ${\frak G}_U$ is an ideal in ${\frak F}_U^{\rm LH}$, derive vertically (\ref{eqn:gauge-d_h}) and notice that $\sigma^X=-{\sf d_v}\rho^X$ is null along ${\partial U}$ thanks to Lemma \ref{lma:X=drho}, in~particular $\sigma^X\vert_{\partial U}$ is ${\sf d_h}$-exact.

We claim that ${\frak G}_U$ is an ideal of $\hat{\frak G}_U$. For~if $X\in{\frak G}_U,{X}'\in {\frak G}_U$ then ${\sf j}[X,X']=[{\sf j}X,{\sf j}X']\vert_{\partial U}$ vanishes.

Finally, to~see that ${\frak F}^{\rm LH}_U\cap \hat{\frak G}_U\subseteq {\frak F}_U^{\rm LH}$ is an ideal, ${\sf d_v}\iota_{{\sf j}[X,V]}\Omega_L$ is ${\sf d_h}$-exact by (\ref{eqn:basic}).  \end{proof}



\section{Linear~Theory}\label{sec:linear}


Recall that each fiber of $\pi:Y\rightarrow M$ is an affine bundle modeled over a linear bundle $\pi^{\sf L}:Y^{\sf L}\rightarrow M$ with $Y^{\sf L}\subseteq {\sf \Omega}^1(M)$. 

Since the space of YM solutions ${\cal A}_U$ is an affine space, take a fixed connection $\eta_0\in {\cal A}_U$, then $\gf=\eta-\eta_0\in \Gamma(Y^{\sf L}\vert_U)$ is such that $d\star d\gf=0$. Here $\star $ denotes the Hodge star operator. In~addition, there exists $V_\gf\in{\frak F}_U,$ such that
\begin{equation}\label{eqn:V_phi}
	{\sf j}V_\gf
	=
	{\sf j}\left(
		\sum_{i=1}^n\gf^i(x) \frac{\partial }{\partial A^i}
	\right)
\end{equation}

Even though Equation~(\ref{eqn:linearized-EL}) imposes a condition {\em on-shell}, i.e.,~on  ${\cal E}_L$ for $V\in{\frak F}_U$, the~linearized equations, $d\star d\gf=0$, induce $V_\gf\in{\frak F}_U$ that satisfies (\ref{eqn:linearized-EL}) {\em off-shell}, that is in ${\sf J}Y$.

As a complementary definition to (\ref{eqn:V_phi}) we may define for every solution, $\eta\in{\cal A}_U$, and~every first variation of solutions, $V\in{\frak F}_U$ the section
\begin{equation}\label{eqn:eta_V}
	\eta_V
	=
	{\sf j}\eta^*V \in \Gamma\left(Y^{\sf L}\vert_U\right)\subseteq  {\sf \Omega}^1(U).
\end{equation}

Here we use the isomorphism, depending on a fixed connection, $\eta_0\in{\cal A}_U$, between~the pullback ${\sf j}\eta_0^*(Y^{\sf v})$ of the vertical bundle $\pi^{\sf v}:Y^{\sf v}\rightarrow {\sf J}Y\vert_U$, and~the linear bundle $\pi^{\sf L}\vert_U:Y^{\sf L}\rightarrow U$.

For the previous definitions the following properties hold
\[
	{\sf j}\eta^*V_{\eta_V}={\sf j}\eta^*V,\qquad
	\eta_{(V_{\gf})}=\gf.
\]
The following assertion holds as an observation that will follow from Lemma~\ref{lma:H-phi}.

\begin{lemma} We have that $V_\gf\in{\frak F}^{\rm LH}_U$ for every $\gf\in{\sf \Omega}^1(U),$ solution of the linearized equation $d\star d \gf=0$. Hence, $V_{\eta_V}\in{\frak F}^{\rm LH}_U$ for every $V\in{\frak F}_U$.
\end{lemma}

The following assertion holds for linear~theories.

\begin{lemma}

For every solution, $\eta\in{\cal A}_U$, and~every first variation of solutions, $V\in{\frak F}_U$, in~a linear theory, there exists ${\gf }\in \Gamma(Y^{\sf L}\vert_U)$ such that
$
	 V\vert_{{\sf j} \eta(U)}=V_\gf\vert_{{\sf j} \eta(U)}
$
or equivalently
$
	 \eta_V=\gf.
$

\end{lemma}

If we want to consider the gauge classes on ${\cal A}_U$ we can consider the gauge representatives consisting of {\em Lorentz gauge fixing conditions}, i.e.,~for every $\eta=\eta_0+\gf \in {\cal A}_U$ there exists a gauge related
\begin{equation}\label{eqn:Lorentz-gauge}
	\tilde{\eta}=\eta_0+\tilde{\gf}\in {\cal A}_U,\quad
	d\star\tilde{\gf}=0
\end{equation}
where $\tilde{\eta}-\eta=\tilde{\gf}-\gf\in{\cal G}_U$ being a gauge translation by exact $1$-forms in ${\cal A}_U$.

Recall the Hodge-Morrey-Friedrichs $L^2$-ortogonal decomposition, see~\cite{Sc}. For~null normal components we have,
\begin{gather}\label{eqn:HMF}
	{\sf \Omega}^1(U)=d{\sf \Omega}_D^{0}(U)\oplus\mathfrak{H}^1_N(U)\oplus
	\left(\mathfrak{H}^1(U)\cap d{\sf \Omega}^{0}(U)\right)\oplus d^{\star}{\sf \Omega}_N^{2}(U)
\end{gather}
where
\[
\begin{array}{rcl}
	{\sf \Omega}^1_N(U)
		&:=&
		\left\{\beta
			\,:	\, \beta\in {\sf \Omega}^{1}(U)\,:\,i^*_{\partial U}\left(\star\beta\right)=0\,\right\},	
		\\
	\mathfrak{H}_N^1(U)
		&:=&
		\mathfrak{H}^1(U)\cap {\sf \Omega}^k_N(U).
		\\
\end{array}
\]

Given a fixed point, $\eta_0\in{\cal A}_U$, the~linear space of Lorentz gauge fixing, $d\star \gf=0$, defines a \mbox{linear subspace} \[
	{\cal L}_U\subseteq 
	\mathfrak{H}^1_N(U)\oplus
	d^{\star}{\sf \Omega}_N^{2}(U)
\]
of linearized solutions, $d^\star d\gf=0$, such that there is a covering, ${\frak e}_{\eta_0}(\gf)=[\eta_0+\gf]$,
\begin{equation}\label{eqn:cover-e}
	{\frak e}_{\eta_0}:{\cal L}_U\rightarrow {\cal A}_U/{\cal G}_U
\end{equation}
of the $[\eta_0]$-component{\em space of solutions modulo gauge,} ${\cal A}_U/{\cal G}_U$.

The following results of this section recover the usual characterizations of gauge symmetries in ${\frak G}_U$ as translations by exact~forms.

\begin{lemma}\label{lma:eta_X-closed}

For every $X\in\hat{\frak G}_U$ and $\eta\in {\cal A}_U$, $d\eta_X=0$.

\end{lemma}
\begin{proof}

If we calculate the square of the $L_2$-norm, $\|d\eta_X\|_2^2=
	\int_U d\eta_X\wedge \star d\eta_X 
,$ of $d\eta_X$ where $\star$ stands for the Hodge star operator for the Riemannian metric $g$, then we get
\[
	\int_U\eta_X\wedge\star d^\star d\eta_X
	-
	\int_{\partial U} \eta_X\wedge \star d\eta_X
\]

If $X\in{\frak F}_U$, $d^\star d\eta_X=0$ then due to Lemma~\ref{lma:X=drho}, the~norm $	\|d\eta_X\|_2^2$ can be calculated as
\[
	-
	\int_{\partial U} \eta_X\wedge \star d\eta_X
	=
		-
	\int_{\partial U} {\sf j}\eta^* (\rho^X) \wedge   d\eta_X
	=
	\int_{\partial U} d\left({\sf j}\eta^* ( \rho^X)\right) \wedge   \eta_X.
\]

Recall (\ref{eqn:gauge-d_h}) and that ${\sf j}\eta^* ({\sf d_h} \rho^X)={\sf j}\eta^*(\iota_{{\sf j} X}\Omega_L)$. Hence
\[
	d\left({\sf j}\eta^* ( \rho^X)\right) \wedge   \eta_X
	=
	{\sf j}\eta^*(\iota_{{\sf j} X}\iota_{{\sf j} X}\Omega_L)=0
\]

Therefore $d\eta_X=0$.
\end{proof}

\begin{proposition}\label{prop:X-exact}

For every solution, $\eta\in {\cal A}_U,$ and every gauge first variation with null boundary condition, ${ X}\in{\frak G}_U$, the~induced $1$-form in the base,
$
	\eta_X$, defined as in (\ref{eqn:eta_V}),
 is exact. Therefore, $\eta_X\in {\cal G}_U$.
\end{proposition}

\begin{proof}
We solve the Poisson BVP for $\psi:U\rightarrow\mathbb{R}$ with Dirichlet boundary conditions
\[\left\{\begin{array}{rcll}
	\Delta \psi&=&
		d^\star \eta_X,&\text{in }U,\\
	\psi\vert_{\partial U}&=&
		0,&$\text{in }$\partial U.
\end{array}\right.\]

Notice that the necessary integral condition for the Poisson equation $\int_U d \star \eta_X d\nu=0$ follows from the boundary condition $\eta_X\vert_{\partial U}=0$.

Thus $\tilde{\eta}_X=\eta_X-d\psi$ is a solution of $d^\star d\tilde{\eta}_X=0$ with Lorentz gauge fixing condition $d^\star \tilde{\eta}_X=0$ and Dirichlet boundary~condition.

Recall (\ref{eqn:HMF}). Since ${\eta}_X\in\hat{\frak G}_U$, according to Lemma \ref{lma:eta_X-closed}, $d{\eta}_X=0$ and $d\tilde{\eta}_X=0$. 

There are two cases: 

Case 1. The~normal component $\partial\psi/\partial x^n\vert_{\partial U}$ does not vanish. Here in local coordinates, $\partial U=\{x^n=0\}$. Then $\tilde{\eta}_X$ is harmonic $(d \tilde{\eta}_X=0$ and $d^\star\tilde{\eta}_X=0$). Therefore, it belongs to $\mathfrak{H}^1(U)\cap d{\sf \Omega}^{0}(U),$  i.e.,~it is~exact.

Case 2. $\partial\psi/\partial x^n\vert_{\partial U}=0$, that is, $\tilde{\eta}_X\in{\sf \Omega}_N^1(U)$. Then $\tilde{\eta}_X\in \mathfrak{H}_N^1(U)\cap \mathfrak{H}_D^1(U)$, i.e.,~$\tilde{\eta}_X=0,$ where
\[
\begin{array}{rcl}
	{\sf \Omega}^1_D(U)
		&:=&
		\left\{\beta
			\,:	\, \beta\in
			{\sf \Omega}^{1}(U)\,:\,\beta(\xi)=0,\,\xi\in {\frak X}(\partial U)\,\right\}	,
		\\
	\mathfrak{H}_D^1(U)
		&:=&
		\mathfrak{H}^1(U)\cap {\sf \Omega}^k_D(U).
		\\
\end{array}
\]

In any case $\tilde{\eta}_X$ is exact and so is $\eta_X$.
\end{proof}

\begin{proposition}\label{prop:X-exact_2}
Take any solution $\eta$, and~any gauge symmetry, ${ X}\in\frak{F}_U^{\rm LH}\cap \hat{\frak G}_U$. Then there exists $X'\in\hat{\frak G}_U$ such that $X-X'\in {\frak G}_U$. Hence $\eta_{X-X'}\in {\cal G}_U$ is exact.
\end{proposition}

\begin{proof}
Take $X\in{\frak F}^{\rm LH}_U\cap \hat{\frak G}_U$. According to the argument given in Proposition \ref{prop:X-exact} we just need to show that the pullback $i^*_{\partial U} \eta_X\in{\sf \Omega}^1(\partial M)$ is null for the inclusion $i_{\partial U}:\partial U\rightarrow U$. Then $\eta_X-d\psi$ would have null Dirichlet condition and would be exact for suitable $\psi$. 

Notice that the following boundary conditions are in general different objects:
\begin{equation}\label{eqn:j|s}
	i^*_{\partial U} \eta_X,
	\qquad
	\eta_X\vert_{\partial U},
	\qquad
	({\sf j}X)\vert_{\partial U},
	\qquad
	{\sf j}\vert_{\partial U}(X\vert_{\partial U})
	.
\end{equation}

Since $X\in \frak{F}^{\rm LH}_U$, then we are assuming a boundary condition on $X$, namely ${\sf d_v}\rho\vert_{\partial U}=\lambda $, with~${\sf d_h}\lambda=0$, when evaluating in ${{\cal E}_L,{\frak F}_U}$. Due to Lemma~\ref{lma:X=drho} we have that $X\vert_{\partial U}$ does not depend on vertical coordinates, $u^j$ when evaluating in ${\cal E}_L$.

We claim that indeed $i^*_{\partial U} \eta_X =0$. Recall that, according to Lemma \ref{lma:X=drho}, for~every $W\in {\frak F}_U$ we have
\[
	\int_{\partial U} {\sf j}\eta^* \left(
		\sum_{j=1}^{n-1}
			\varpi^{nj} X^j
			\left(\frac{{\sf d} W^j}{{\sf d} x^n}-\frac{{\sf d} W^n}{{\sf d} x^j}\right)
		\nu^n
	\right)=
\]\[
	\int_{\partial U} {\sf j}\eta^* \left(
		\iota_{{\sf j}W} \iota_{{\sf j}X}\Omega_L
	\right)
	= 
	\int_{\partial U}  
		{\sf j}\eta^*( \iota_{{\sf j} W}{\sf d_h}\rho)
	=
	\int_{\partial U}  
		d\left({\sf j}\eta^*( \iota_{{\sf j} W} \rho)\right)=0.
\]

Therefore, $X^j({\sf j}\eta)\vert_{\partial U}=0$ for $j=1,\dots, n-1$, hence null Dirichlet boundary conditions hold for $\eta_X$. There exists a smooth function $f:U\rightarrow\mathbb{R}$ such that $i^*_{\partial U}df=i_{\partial U}^*\eta_X=0$, and~$\frac{\partial f}{\partial x^n}\vert_{\partial U} =X^n({\sf j}\eta)$. If~\[
	{\sf j}X_{df}=
	{\sf j}\left(\sum_{i=1}^n
		\frac{\partial f}{ \partial x^i}\frac{\partial}{\partial u^i}
	\right),
\]
then ${X}_\eta':={\sf j}\eta^*(X-X_{df})={\eta}_X -df$ has null both Neumann and Dirichlet conditions on $\partial U$. We just need to refine the choice of $f$, so that  ${\sf j}{X}'\vert_{\partial U}=0$. Hence ${X}'\in{\frak G}_U$.
\end{proof}

\begin{theorem}\label{thm:inclusion}
There is an inclusion of the gauge quotients of Lie algebras, 
\[\xymatrix{
	{\frak F}_U^{\rm LH}/{\frak G}_U
		\ar@{^{(}->}[r]^{\hat{\cdot}}
		&
	\hat{\frak F}_U^{\rm LH}/\hat{\frak G}_U
}.\]
\end{theorem}

\begin{proof}
By the Second Isomorphism Theorem for Lie algebras
\[
	\hat{\frak F}^{\rm LH}_U/\hat{\frak G}_U
	\simeq
	\left(\hat{\frak F}^{\rm LH}_U/{\frak G}_U\right)
	/
	\left(\hat{\frak G}_U/{\frak G}_U\right).
\]

Notice that 
\[
	{\frak G}_U\subseteq
	 {\frak F}_U^{\rm LH}\cap \hat{\frak G}_U\subseteq 
	\ker \Psi 
\]where $\Psi$ is the Lie algebra morphism defined as the composition in the diagram below.
\[\xymatrix{
	&
	{\frak F}^{\rm LH}_U/{\frak G}_U
	&
	\hat{\frak F}^{\rm LH}_U/\hat{\frak G}_U
		\ar@{<->}[r]
	&
	\left(\hat{\frak F}^{\rm LH}_U/{\frak G}_U\right)
	/
	\left(\hat{\frak G}_U/{\frak G}_U\right)
	\\
	&
	{\frak F}^{\rm LH}_U
			\ar[ru]_\Psi
			\ar@{->>}[u]
	\ar@{^{(}->}[r]
	&
	\hat{\frak F}^{\rm LH}_U
	\ar@{->>}[r]
	&
	\hat{\frak F}^{\rm LH}_U/{\frak G}_U
	\ar@{->>}[u]
	\\
	\ker \Psi
	\ar@{^{(}->}[ur]
	&
	{\frak G}_U
		\ar@{^{(}->}[l]
			\ar@{^{(}->}[u]
	&&
}\]

By the first isomorphism theorem, there exists an induced monomorphism $\tilde{\Psi}$ and a \mbox{commutative diagram}
\[\xymatrix{
	{\frak F}^{\rm LH}_U /	{\frak G}_U	
			\ar@{^{(}->}[d]
	&
	{\frak F}^{\rm LH}_U /	\ker \Psi	
			\ar@{^{(}-->}[d]^{\tilde{\Psi}}
			\ar@{->>}[l]
	\\
	\hat{\frak F}^{\rm LH}_U /	{\frak G}_U
	&
	\hat{\frak F}^{\rm LH}_U /	\hat{\frak G}_U.
		\ar@{->>}[l]
}\]

There is an inclusion ${\frak F}^{\rm LH}_U\cap \hat{\frak G}_U\subseteq {\frak F}_U^{\rm LH}$. Hence $\ker \Psi = {\frak F}^{\rm LH}_U\cap \hat{\frak G}_U$. By~Proposition \ref{prop:X-exact_2}, the~inclusion ${\frak G}_U\subseteq {\frak F}^{\rm LH}_U\cap\hat{\frak G}_U,$ is a section of the projection 
$
	\xymatrix{
	{\frak F}^{\rm LH}_U\cap \hat{\frak G}_U
		\ar@{-->>}[r]
	&
	{\frak G}_U,
}$ given by $X\mapsto X-X'$.

Therefore, we have the required inclusion
\[\xymatrix{
	{\frak F}^{\rm LH}_U /	{\frak G}_U	
			\ar@{^{(}->}[d]
			\ar@{^{(}-->}[r]
					\ar@{^{(}-->}[dr]		
	&
	{\frak F}^{\rm LH}_U /	\ker \Psi	
			\ar@{^{(}->}[d]^{\tilde{\Psi}}
	\\
	\hat{\frak F}^{\rm LH}_U /	{\frak G}_U
	&
	\hat{\frak F}^{\rm LH}_U /	\hat{\frak G}_U.
		\ar@{->>}[l]
}\]\end{proof}

Recall that $ H^1_{\rm dR}(U,\partial U)\simeq {\frak H}^1_D(U)$ in the exact sequence,
\begin{equation}\label{eqn:exactseq}
\xymatrix{
	H^1_{\rm dR}(U,\partial U)
		\ar[r]
		&
	H^1_{\rm dR}(U)
		\ar[r]^{i^*_{\partial U}}
		&
	H^1_{\rm dR}(\partial U).	
}\end{equation}

Hence, the~demand in the proof of Proposition \ref{prop:X-exact_2} for $i^*_{\partial U} \eta_X$ to be null is equivalent to demanding $\eta_X$ to lie into ${\sf \Omega}^1_{D} (U)$. Thus, $\eta_X$ defines a relative cohomology class $[\eta_X]\in H^1_{\rm dR}(U,\partial U)$. Further considerations actually explain that $[\eta_X]=0$.

\begin{proposition}\label{prop:rel-coh}

If $H^1_{\rm dR}(U,\partial U)=0$, then ${\frak F}_U^{\rm LH}/ {\frak G}_U\simeq \hat{\frak F}_U^{\rm LH}/ \hat{\frak G}_U$.

\end{proposition}

\begin{proof}
For every $V\in{\frak F}^{\rm LH}_U$ we have that 
$
	{\sf d_v}\iota_{{\sf j} V}\Omega_L \vert_{\partial U}
	=
	{\sf d_h}\sigma^V,$ with $\sigma^V\vert_{\partial U}={\sf d_h}\lambda\vert_{\partial U}$.
Take $\eta\in {\cal A}_U$ any YM solution. For~$\eta_V\vert_{\partial U}={\sf j}\eta^*V\vert_{\partial U}$, we solve the Poisson BVP
\[
	\left\{
		\begin{array}{ll}
		\Delta\psi = d^\star \eta_V,
		&
		\text{\ in }U,
		\\
		\partial \psi/\partial x^n\vert_{\partial U}=- V^n(x),
		&
		\text{ in } \partial U=\{x^n=0\},
		\end{array}
	\right.
\] 
then $\eta_V$ may be gauge translated by an exact form $d \psi$ so that $\tilde{\eta}_V=\eta_V+d\psi$ has no normal components along $\partial U$ and satisfies $d\star \tilde{\eta}_V=0$ as well as the linearized YM equation, $d^\star d\tilde{\eta}_V=0$. 

Notice that the induced linearized solution $X_{d\psi}\in{\frak F}_U$ in fact belongs to $\hat{\frak G}_U\cap{\frak F}^{\rm LH}_U$.

By (\ref{eqn:HMF}) $\tilde{\eta}_V\in {\frak H}_N^1(U)\oplus ({\frak H}^1(U)\cap d{\sf \Omega}^0(U))\oplus d^\star{\sf \Omega}^2_N(U)$. For~the coclosed projection $\eta_V'\in {\frak H}_N^1(U)\oplus d^\star{\sf \Omega}^2_N(U)$ of $\tilde{\eta}_V$, we have the orthogonal decomposition, ${\eta}'_V=\eta_V''\oplus d^\star \chi$.

Consider the {\em boundary conditions linear map}, $
	{\rm d}r_{\partial U}[\eta]:{\frak F}^{\rm LH}_U/{\frak G}_U\rightarrow 
	{\cal L}_{\partial U}$, such that 
\[
	({\rm d}r_{\partial U})[\eta](V)
	=
	(\eta_V')^D\oplus(\eta_V')^N
	:=
	[\iota_{\partial U}^* {\eta}'_V]\oplus \star_{\partial U}i^*_{\partial U}(\star d{\eta}'_V),\]
where the codomain is the linear space of Dirichlet-Neumann boundary conditions modulo gauge, 
\begin{equation}\label{eqn:L_partialU}
	{\cal L}_{\partial U}:=\left(\ker d^{\star_{\partial U}}/(d{\sf \Omega}^0(U)\cap \ker d^{\star_{\partial U}})\right)\oplus\ker d^{\star_{\partial U}}.
\end{equation}
See~\cite{DM-O} for further considerations of this space. Recall the isomorphisms 
\[
	{\frak H}_N^1(U)\simeq H^1_{\rm dR}(U),
	\quad
	{\frak H}_D^1(U)\simeq H^1_{\rm dR}(U,\partial U).
\]

Since ${\frak H}_D^1(U)\simeq H^1_{\rm dR}(U,\partial U)=0,$ then by (\ref{eqn:exactseq}) we have  ${\frak H}_N^1(U)\subseteq H^1_{\rm dR}(\partial U) $. Hence, the~closed projection of $i_{\partial U}^*(\eta_V')\in {\sf \Omega}^1(\partial U)$ would have cohomology class $i^*_{\partial U}[\eta_V'']$ in  $H^1_{\rm dR}(\partial U)$ induced by $[\eta_V'']\in H^1_{\rm dR}(U) $. Therefore, $({\rm d}r_{\partial U})[\eta]$ is~injective.

If we proceed as in the previous argument with $[V]\in\hat{\frak F}^{\rm LH}/\hat{G}_U,$ we can define an injective map $\widehat{\rm d}r_{\partial U}$ such that the following diagram commutes
\[\xymatrix{
	{\frak F}^{\rm LH}_U/{\frak G}_U
		\ar@{^{(}->}[r]^{	{\rm d}r_{\partial U}[\eta]}
		\ar@{^{(}->}[d]^{\check{\eta}}
	&  
	{\cal L}_{\partial U}
	\\
	\hat{\frak F}^{\rm LH}_U/\hat{\frak G}_U.
		\ar[ur]_{\widehat{\rm d}r_{\partial U}[\eta]}
		&
}\]

Notice that $\widehat{\rm d}r_{\partial U}[\eta]$ and ${\rm d}r_{\partial U}[\eta]$ have the same image.
\end{proof}

Remark that we have the commutative diagram
\begin{equation}\label{eqn:diagram3}
\xymatrix{
	{\cal A}_U/{\cal G}_U
&
	{\cal L}_U
		\ar[r]^{r_{U,\partial U}}
		\ar[l]^{{\frak e}_\eta}
&
	{\cal L}_{\tilde{U}}.
\\
	&
	{\frak F}^{\rm LH}_U/{\frak G}_U
		\ar[ru]_{{\rm d}r_{\partial U}[\eta]}
		\ar[u]^{\mathfrak{p}_\eta}
		\ar@{-->}[ul]^{\mathfrak{exp}_\eta}
	&
}\end{equation}
where
\[
	{\cal L}_{\tilde{U}}:=r_{U,\partial U}({\cal L}_U)\subseteq {\cal L}_{\partial U}
\]
with $r_{U,\partial U}$ the map of {\em boundary conditions of solutions modulo gauge}, see~\cite{DMO} for further properties of this map. Here we use axial gauge fixing in a tubular neighborhood of $\partial U$ as well as the linear map $r_{U,\partial U}(\gf)=\gf^D\oplus\gf^N$ is defined in (\ref{eqn:L_partialU}). The~linear map ${\frak p}$ is induced by ${\frak p}_\eta(V)=\eta+\gf$ where $\gf\in{\sf \Omega}^1(U)$ is a coclosed linearized solution, $d\star d\gf=0$ such that $\eta_V=\gf$, see notation (\ref{eqn:eta_V}). 

By composing the projection ${\frak p}_\eta$ with the map ${\frak e}_\eta$ we get the map $\mathfrak{exp}_\eta:	{\frak F}^{\rm LH}_U/{\frak G}_U\rightarrow {\cal A}_U/{\cal G}_U$. Diagram (\ref{eqn:diagram3}) suggests that {\em Hamiltonian first variation modulo gauge}, $	{\frak F}^{\rm LH}_U/{\frak G}_U$ is a Lie algebra isomorphic as linear space to the tangent space of the moduli space ${\cal A}_U/{\cal G}_U$ at $\eta$.

The following assertion related to Proposition \ref{prop:rel-coh} explains how the relative cohomology codifies the description of ${\cal A}_U/{\cal G}_U$ with respect to the boundary~conditions, see also~\cite{DM-O}.

\begin{proposition}\label{pro:fibration}

$H^1_{\rm dR}(U,\partial U)=0$ if and only if $r_{\partial U}:{\cal L}_U\rightarrow {\cal L}_{\partial U}$ is injective and $r_{\partial U} :{\cal L}_U\rightarrow {\cal L}_{\tilde{U}}$ is a linear~isomorphism.

\end{proposition}


\section{Poisson-Lie Algebra of Hamiltonian~Observables}\label{sec:lie}



\begin{definition}[Hamiltonian observable currents]
We say that an observable current $F\in {\sf \Omega}^{n-1,0}({\sf J}Y\vert_U)$ is a {\em Hamiltonian observable current} if there exist $V\in {\frak F}_U$ and a residual form $\sigma^F$ such that the following relation holds when restricted to ${\cal E}_L$ and evaluated on $W\in{\frak F}_U$,
\begin{equation}\label{eqn:HOC}
	{\sf d_v} F\vert_{{\cal E}_L}=-\iota_{{\sf j} V}\Omega_L+{\sf d_h}\sigma^F
\end{equation}

We denote the space of Hamiltonian observable currents over $U$ as $\widehat{\rm HOC}_U$. The evolutionary vector field ${ V}$, is actually a {locally Hamiltonian first variation}, i.e.,~$V\in \hat{\frak F}^{\rm LH}_U$. If~in addition in (\ref{eqn:HOC}) we have the \mbox{boundary condition}
\begin{equation}\label{eqn:HOC-bdry}
	{\sf d_v}\sigma^F\vert_{\partial U}={\sf d_h}\lambda^F
\end{equation}
then we call $F$ a Hamiltonian observable current {\em with boundary condition}. Here $V\in{\frak F}^{\rm LH}_U$. We denote the space of these kind of observable currents as ${\sf HOC}_U$.
\end{definition}

\begin{definition}[Helicity current]
Suppose that $\gf\in{\sf \Omega}^1(U)$ is a solution of the linearized YM equation, $d \star d \gf=0$. Define the $\gf$-{\em helicity current} as
\[
	F^{\gf}=\iota_{{\sf j}V_\gf} \iota_{{\sf j}R}\Omega_L\in{\sf \Omega}^{n-1,0}({\sf J}Y\vert_U),
\]
where $R\in{\frak F}_U$ was defined in (\ref{eqn:R}). More explicitly
\[
	F^\gf=\sum_{i=1}^n\left[
		\sum_{j=1}^n\varpi^{ij}\left(
				\gf^j(x)\cdot (A^j_i-A^i_j)
			-	
				A^j\cdot
				\left(\frac{\partial \gf^j(x)}{{\partial} x^i}
				-\frac{{\partial} \gf^i(x)}{{\partial} x^j}
				\right)
		\right)
	\right]\nu_i.
\]
\end{definition}

Form the very definition and the multysimplectic formula it can be seen that ${\sf d_h}F^\gf\vert_{{\cal E}_L}=0$.

Remark that we could have defined observable currents, $F^\gf$, for~{\em any} divergence-free ${\gf}$ in $U$, $d\star \gf=0$, with~evolutionary Hamiltonian vector field, $V_\gf\in\mathfrak{Ev}({\sf J}Y\vert_U)$, rather than in restricting ourselves to Hamiltonians first variations in ${\frak F}_U$, just as the observables considered in~\cite{AB}. Nevertheless, if~we had adopted this definition, then we would have to restrict the domain of $F^\gf$ and evaluate only ob solutions $\eta'=\eta_0+\gf'\in{\cal A}_U$ with Lorentz gauge fixing (\ref{eqn:Lorentz-gauge}), $\gf'\in{\cal L}_U$ in order to have local invariance ${\sf d_h}F^\gf\vert_{{\cal E}_L}=0$.

From the following assertion it follows that helicity currents are Hamiltonian observable currents restricted to $U$, that is $F^\gf\in{\sf HOC}_U$. 

\begin{lemma}\label{lma:H-phi}
The $\gf $-helicity current, 
$
	F^\gf \in {\sf \Omega}^{n-1,0}({\sf J}Y\vert_U,$
defines a locally Hamiltonian observable current with Hamiltonian $V_\gf\in {\frak F}_U^{\rm LH}$ whenever $d\star d\gf=0$.
\end{lemma}

\begin{proof}
Recall the notation in (\ref{eqn:V_phi}). Notice that the relation ${\sf d_v} F^\gf + \iota_{{\sf j}V_\gf}\Omega_L=0$ is valid off-shell. Therefore we have
\[
	{\sf d_v} F^\gf\vert_{{\cal E}_L} 
	=
	- \iota_{{\sf j}V_\gf}\Omega_L 
\]
in particular when evaluated on $W\in{\frak F}_U$.	
\end{proof}

\begin{lemma}\label{lma:sympl-product}
If $\gf,\gf'\in {\sf \Omega}^1(U)$ are solutions of $d\star d\gf'=0=d\star d\gf$, then the Lie derivative,
$
 	\mathscr{L}_{{\sf j} V_{{\gf'}}} F^\gf, 
$
lies in ${\sf HOC}_U$ with Hamiltonian $[V_\gf,V_{\gf'}]\in {\mathfrak{F}_U}$. Under~integration over $\Sigma$, it yields the {\em symplectic product observable}, associated to $\iota_{{\sf j}V_\gf}\iota_{{\sf j}V_{\gf'}}\Omega_L\in {\sf HOC}_U$,
\[
	f_\Sigma^{V_\gf V_{\gf'}}(\eta)
	:=
	\int_\Sigma {\sf j}\eta^*\iota_{{\sf j}V_\gf}\iota_{{\sf j}V_{\gf'}}\Omega_L,
	\quad
	\forall \eta \in{\cal A}_U,\phi\in {\cal L}_U .
\]
\end{lemma}

\begin{proof}

Notice that
\[
	\mathscr{L}_{{\sf j} V_{{\gf'}}} F^\gf
	=
	\iota_{{\sf j}V_{\gf'}} {\sf d_v}  F^\gf
	=
	-\iota_{{\sf j}V_{\gf'}}\iota_{{\sf j}V_{\gf}}\Omega_L
\]
evaluated on $W\in{\frak F}_U$ on Shell. On~the other hand a general formula (\ref{eqn:basic}) states that
\[
	{\sf d_v} \left( \iota_{{\sf j}V_{\gf}}\iota_{{\sf j}V_{\gf'}}\Omega_L\right)
	=
	-\iota_{[{\sf j}V_{\gf},{\sf j}V_{\gf'}]} \Omega_L
	+
	\iota_{{\sf j}V_{\gf'}} \mathscr{L}_{{\sf j}V_{\gf}}\Omega_L
	-
	\iota_{{\sf j}V_{\gf}} \mathscr{L}_{{\sf j}V_{\gf'}}\Omega_L.
\]

Therefore $	{\sf d_v} \left( \iota_{{\sf j}V_{\gf}}\iota_{{\sf j}V_{\gf'}}\Omega_L\right) =
	-\iota_{[{\sf j}V_{\gf},{\sf j}V_{\gf'}]} \Omega_L.$
Recall that ${[{\sf j}V_{\gf},{\sf j}V_{\gf'}]}={{\sf j}[V_{\gf},V_{\gf'}]}$, see for instance~\cite{Vinogradov} por the explicit form of the Lie bracket of evolutionary vector fields. Hence $[V_{\gf},V_{\gf'}]$ is Hamiltonian first variation for $\iota_{{\sf j}V_{\gf}}\iota_{{\sf j}V_{\gf'}}\Omega_L\in {\sf HOC}_U$.
\end{proof}

Define the family of $\gf$-{\em helicity observables} as
\[
	f^\gf_\Sigma(\eta)=\int_\Sigma ({\sf j}\eta)^* F^\gf,
	\qquad
	\forall \eta\in{\cal A}_U.
\]

We see that $f^{\gf}_\Sigma$ is related to the anti-symmetric component of the helicity as bilinear form, see Section~\ref{hel}, in~the sense of (\ref{eqn:hel2}). 
Notice also that $[\cdot,\cdot]_\Sigma$ is not necessarily symmetric, unless~$\partial \Sigma=0$. Hence $f^\gf_\Sigma$ not necessarily equals $0$.

We say that $f_\Sigma^\gf$ is a {\em Hamiltonian observable} with {\em Hamiltonian  first variation} $v_\gf$ so that the following formal identity holds:
\begin{equation}\label{eqn:notation-Lie}
	\mathscr{L}_w f_\Sigma^\gf(\eta)
		=
		-\omega_{\Sigma L}[\eta](v_\gf,w)
		, \qquad
		\forall w=\delta \eta,\forall \eta\in{\cal A}_U.
\end{equation}

Let us explain the formal notation of (\ref{eqn:notation-Lie}). Any first variation of solutions, $W\in {\frak F}_U,$ encodes a variation of any fixed solution $\eta\in{\cal A}_U$, which we denote as $w=\delta \eta$,
\begin{equation}\label{eqn:w}
	w=\left. \frac{{\rm d}\phi^\geps}{{\rm d}\geps}\right\vert_{\geps=0}
	,\qquad \phi^0=\eta
\end{equation}
for a one-parameter family of smooth solutions ${\phi^\geps}\in{\cal A}_U.$ This means that
$
		\left. \frac{{\rm d}({\sf j}\eta^\geps)}{{\rm d}\geps}\right\vert_{\geps=0}
		  = {\sf j}W({\sf j} \eta)
$.

In the r.h.s. we have an evaluation of a symplectic form,
\begin{equation}\label{eqn:omega}
	\omega_{\Sigma L}[\eta](v,w):=
	\int_\Sigma({\sf j}\eta)^*\Omega_L({\sf j} V,{\sf j}W).
\end{equation}

While in the l.h.s. we have
\begin{equation}\label{eqn:lie}
	\mathscr{L}_w f_\Sigma^\gf(\eta)=\left. \frac{ d}{{ d}\geps}\right\vert_{\geps=0}
		f^\gf_\Sigma(\phi^\geps),
\end{equation}

With this notation we suggest that we are modeling a Lie derivative ${\mathscr L}_w(\cdot)$ in the tangent space of the moduli space ${\cal A}_U/{\cal G}_U$, while $w=\delta\eta$ corresponds to local vector fields near $[\eta]\in {\cal A}_U/{\cal G}_U$.

If $X\in {\frak G}_U$ corresponds to a first variation of a one-parametric family of gauge equivalent solutions, $\phi^\geps$, then
$
	\mathscr{L}_xf^{\gf}_\Sigma= 0,
$
which follows from ${\sf j} X\vert_{\partial U}=0$. Thus $f^\gf_\Sigma$ is well defined for the gauge class  $[V_\gf]\in \mathfrak{F}^{\rm LH}_U/{\frak G}_U$.

\begin{lemma}

Consider the linear space
\[
	{\frak f}_{\Sigma U}:=
	\left\{
		f^\gf_\Sigma
		\,:\,
		d\star d \gf =0,\, 
		[V_\gf]\in \mathfrak{F}^{\rm LH}_U/{\frak G}_U
	\right\} /\mathbb{R}
\]
where $f^{\gf_1}_\Sigma- f^{\gf_2}_\Sigma$ is a constant function iff represent the same $\mathbb{R}$-class. Then ${\frak f}_{\Sigma U}$ is a Lie algebra with bracket
\[
	\left\{
		\left[f^\gf_\Sigma\right],
		\left[f^{\gf'}_\Sigma\right]
	\right\}_\Sigma
	=
	\left[f^{V_\gf V_{\gf'}}_\Sigma
	\right],
\]
which means
\[ \left\{
		f^\gf_\Sigma,
		f^{\gf'}_\Sigma
	\right\}_\Sigma
	=
	f^{V_\gf V_{\gf'}}_\Sigma + {\rm const.}
\]

\end{lemma}

\begin{proof}

Let $\gf,\gf'$ be $1$-forms as in the hypothesis. As~in the proof of Lemma \ref{lma:sympl-product}, recall that
\[
	{\sf d_v} \left( \iota_{{\sf j}V_{\gf}}\iota_{{\sf j}V_{\gf'}}\Omega_L\right)
	=
	-\iota_{{\sf j}[V_{\gf},V_{\gf'}]} \Omega_L.
\]

There are gauge translations $X=X_{\psi},X'=X_{\psi'}\in{\frak G}_U, \psi,\psi':U\rightarrow\mathbb{R}$ such that the gauge translations $V, V'$ are divergence-free, see for instance the Appendix~\cite{DM-O}. Recall that $V,V'$ are defined by $\gf-d\psi,{\gf'}-d\psi',$ respectively. Hence $V=V_\gf-X$ and $V'=V_{\gf'}-X'$. By~(\ref{eqn:basic})
\[
	-\iota_{{\sf j}[V_\gf,V_{\gf'}]}\Omega_L
	=
	-\iota_{{\sf j}[V,V']}\Omega_L
	-
	\iota_{{\sf j}[V,X']}\Omega_L-
	\iota_{{\sf j}[X,V']}\Omega_L-
	\iota_{{\sf j}[X,X']}\Omega_L
\]

Hence
\[
	{\sf d_v}\left(\iota_{{\sf j}V_\gf} 	\iota_{{\sf j}V_{\gf'}}\Omega_L\right)
	=
	-\iota_{{\sf j}[V,V']}\Omega_L
	+
	{\sf d_h}\sigma^{VV'}
\]
with 
\[
	\sigma^{VV'}
	=
	-\rho^{[V,X']}-\rho^{[X,V']}-\rho^{[X,X']}.
\]

Denote $\tilde{\gf}\in{\sf \Omega}^1(U)$ as the a $1$-form such that ${[V,V']}=V_{\tilde{\gf}}$. In~local coordinates:
\[
	\tilde{\gf}^j
	=
	\sum_{i=1}^n
	\gf^i_1(x)\frac{d \gf^j_2(x)}{d x^i}-
		\gf_2^i(x)\frac{d \gf_1^j(x)}{d x^i},
		\qquad
		\gf_1=\gf -d\psi,\,
		\gf_2=\gf' -d\psi',\,	 
	.
\]
Recall that divergence-free vector fields form a Lie algebra, that is $d\star \tilde{\gf}=0$. Then
\[
		{\sf d_v}\left(F^{\tilde{\gf}}
	+ 		\iota_{{\sf j}V_{\gf}}\iota_{{\sf j}V_{\gf'}}
		\Omega\right)
	=
	{\sf d_h}\sigma^{VV'}.
\]

Therefore,
\[
 \mathscr{L}_w \left(f^{V_\gf V_{\gf'}}_\Sigma
 	-
	f_\Sigma^{\tilde{\gf}}\right)
 =
0,
\] for every variation of solutions $w$ associated to every $W\in{\frak F}_U$. See the explanation of the notation in (\ref{eqn:notation-Lie}). Hence $f_\Sigma^{V_\gf V_{\gf'}}=f^{\tilde{\gf}}_\Sigma+{\rm const.}$
\end{proof}

We claim that ${\frak f}_{\Sigma U}$ yields a family of local observables sufficiently rich to separate solutions, see also~{\cite{Zapata(2019)}. Suppose that we consider a non-gauge variation $v=\delta\eta$ of a solution $\eta\in {\cal A}_U$. More precisely, take a one-parametric family of solutions $\eta^\geps=\eta+\geps \gf$ encoded by the symmetry $V\in{\frak F}_U$, that is
$
	\left.\frac{\rm d}{{\rm d} \geps}\right\vert_{\geps =0} {\sf j}(\eta^\geps)
	=
	{\sf j}V({\sf j}\eta).
$
Without loss of generality we can also suppose that $V=V_\gf$ with $d\star d\gf=0$.
Hence, for~any $0\neq [V]\in {\frak{F}}^{\rm LH}_U/{\frak G}_U$, there exists $[W]\in{\frak F}^{\rm LH}_U/{\frak G}_U,\, [W]\neq 0,$ such that 
$
	{\sf j}\eta^*\left(\iota_{{\sf j}V}\iota_{{\sf j}W}\Omega_L\right)$
in a suitable open $n$-dimensional ball $U'\subseteq U$. We choose an embedded $(n-1)$-dimensional ball, $\Sigma'\subseteq U',\, \partial \Sigma'\subseteq \partial U'$ such that
\[
	\int_{\Sigma'}{\sf j}\eta^*\left(\iota_{{\sf j}W_{\gf'}}\iota_{{\sf j}V_{\gf}}\Omega_L\right)
	\neq 0
\]
for $W=W_{\gf'}$ associated to $\gf'\in{\sf \Omega}^1(U), d\star d\gf'=0,$ a non trivial solution to linearized equations in $U'$ that also vanishes in the exterior of $U'$. 

We then extend $\Sigma'$ to $\Sigma\subseteq U,\, \partial\Sigma\subseteq\partial U,$ such that $f^{VW}_\Sigma(\eta)\neq 0$. The variation of $f^{\gf}_\Sigma$ along  $w$ in the space of YM solutions is
\[
	\mathscr{L}_w f^{\gf}_\Sigma(\eta)
	=f_\Sigma^{V W}(\eta)\neq 0.
\]

Remark that for every YM solution $\eta\in{\cal A}_U$ and for every variation $V\in {\frak F}_U$, if~$\gf={\sf j}\eta^*V$, then and
$
	\left.
		{\sf j}V_\gf
	\right\vert_{{\sf j}\eta}
	= 
	\left.
	{\sf j}V 
	\right\vert_{{\sf j}\eta}.
$
Thus we could change notation and index the family $\{f^\gf_\Sigma\}$ as $\{f^V_\Sigma\}$ where we take $V$ in ${\frak F}^{\rm LH}_U$.

We summarize the results exposed in this section in the following result and regard the family of observables $\{f^V_\Sigma\}$ as a ``Darboux local coordinate system'' for our gauge field~theory.

\begin{theorem}[Darboux's Theorem]\label{tma:darboux}
Given $\eta\in{\cal A}_U$ a fixed YM solution. For~each $\Sigma\subseteq U$ an admissible hypersurface, $\partial \Sigma\subseteq \partial U$, with~relative homology class $[\Sigma]\in H_{n-1}(U,\partial U)$, there exists an infinite dimensional gauge invariant Lie algebra (modulo constant functions)
\[
	{\mathfrak{f}}_{\Sigma U}=
		\left\{
		f_\Sigma^V\,:\, [V]\in {\frak F}^{\rm LH}_U/{\frak G}_U
	\right\}/\mathbb{R}
\]
such that the following assertions~hold:
\begin{enumerate}

\item ${\frak f}_{\Sigma U}$ is gauge invariant: If $X$ is a variation of one-parametric family of gauge equivalent solutions then
$\mathscr{L}_{x}f^V_\Sigma=0.
$ Moreover, $[f_\Sigma^V]\in {\frak f}_{\Sigma U}$ depends just on the gauge ${\frak G}_U$-class, $[V]\in {\frak F}_U^{\rm LH}/{\frak G}_U$.

\item  Each variation $V$ is in fact locally Hamiltonian, $V\in {\frak F}_U^{\rm LH}$ hence $f^V_\Sigma$ is an observable that satisfies the Hamilton's equation (recall notation {in} (\ref{eqn:notation-Lie})):
\[
		\mathscr{L}_w
		f^V_\Sigma(\eta)
		=
		-\omega_{\Sigma L}[\eta](v,w)
		,\qquad
		\forall 
		w=\delta\eta\]

\item $	{\mathfrak{f}}_{\Sigma U}$, locally separates solutions near $\eta$: For every non-gauge variation $v=\delta\eta$ modeled by $V\in {\frak F}_U$, there exists a locally Hamiltonian variation $w$ modeled by $W\in {\frak F}^{\rm LH}_U$ and $\Sigma\subseteq U$ with
\[
	\mathscr{L}_vf^W_\Sigma(\eta)\neq 0.
\]

\end{enumerate}
\end{theorem}

The following commutative diagram of Lie algebra morphisms and vertical exact sequences summarizes our results
\begin{equation}\label{eqn:diagram1}
\xymatrix{
	{\frak f}_{\Sigma U}
	&
	&
	{\frak F}_U^{\rm LH}/{\frak G}_U
		\ar@{->>}[ll]
		\ar@{^{(}->}[r]
	&
	\hat{\frak F}^{\rm LH}_U/\hat{\frak G}_U
	\\
	\{F^\gf\}
		\ar@{->>}[u]^{\int_\Sigma\cdot}
		\ar@{^{(}->}[rr]
	&
	&
	{\sf HOC}_U
		\ar[u]
		\ar@{^{(}->}[r]
	&
	\widehat{\sf HOC}_U
		\ar[u]
	\\
	\mathbb{R}
		\ar[u]
	&
	&
	{\sf C}
		\ar@{->>}[ll]
		\ar[u]
		\ar@{^{(}->}[r]
	&
	\hat{\sf C}_U
		\ar[u]
}\end{equation}
where ${\sf C}_U$ denote subset of the the constant observable currents
\[
	\hat{\sf C}_U:=
	\left\{
		F\in\widehat{\rm HOC}_U
		\,:\,
		({\sf d_v} F- {\sf d_h}\sigma^F)\vert_{{\frak F}_U,{\cal E}_L}=0
	\right\}
\]
with the additional boundary condition $({\sf d_v} \sigma^F- \lambda^F)\vert_{{\frak F}_U,{\cal E}_L}=0$, $ {\sf d_h} \lambda^F\vert_{{\frak F}_U,{\cal E}_L}=0$.

\begin{definition}[Poisson algebra]
Let $\Sigma$ be any admissible hypersurface $\Sigma\subseteq U$. The~{\em (polynomial) Poisson algebra of helicity Hamiltonian observables}, 
\[
	\left(
		\mathscr{P}\left({\frak f}_{\Sigma U}\right),
		\{\cdot,\cdot\}_\Sigma
	\right)
\]
is generated by the Lie algebra
$
	{\frak f}_{\Sigma U}
	=
	\{[f^\gf_\Sigma]:{\cal A}_U/{\cal G}_U\rightarrow\mathbb{R}\}.
$
\end{definition}

The proof of the following assertion follows from the fact that the space of boundary conditions of solutions, ${\cal L}_{\tilde{U}}\subseteq {\cal L}_{\partial U}$, is a Lagrangian subspace with respect to the symplectic form $\omega_{\partial U ,L}$, see~\cite{CMR1}.

\begin{proposition}
For a hypersurface $\Sigma\subseteq \partial U$ (such that $[\Sigma]=0\in H_{n-1}(U,\partial U))$ and for its complement, $\Sigma'=U- \Sigma\subseteq \partial U,$ the corresponding observables uniquely define an observable
\[
	f_{\partial\Sigma}^\gf:=f_\Sigma^\gf=-f^\gf_{\Sigma'}\in{\frak f}_{\Sigma U}
\]
associated to the oriented and closed $(n-2)$-dimensional boundary $\partial \Sigma\subseteq \partial U$.
\end{proposition}

The Lie algebra
\[
	{\frak f}_{\partial U}:=\{f_{\partial \Sigma}^\gf\,:\, \Sigma\subseteq \partial U\}/\mathbb{R}
\]
will suffice to separate boundary conditions of solutions, while the Lie algebras ${\frak f}_{\Sigma U}$ corresponding to $0\neq [\Sigma]\in H_{n-1}(U,\partial U)$ will be necessary if we want to separate solutions yielding the same boundary conditions, hence in the fibers of ${ r}_{U,\partial U}:{\cal L}_U\rightarrow  {\cal L}_{\tilde{U}} \subseteq {\cal L}_{\partial U}$. This happens when $H^1_{\rm dR}(U,\partial U)\neq 0$ according to Proposition~\ref{pro:fibration}. This also allows us to consider the fibers of ${r}_{U,\partial U}:{\cal L}_U\rightarrow {\cal L}_{\tilde{U}}$ as the symplectic leafs the coisotropic linear space ${\cal L}_U$. This image has been described in detail for the moduli space ${\cal A}_U/{\cal G}_U$ of non-abelian YM solutions in the two dimensional case, see for instance~\cite{Sengupta}.


\section{Gluing Observable~Currents}\label{sec:gluing}

Suppose that a region $U$ is obtained by gluing $U_1,U_2$ along the closed hypersurfaces $\Sigma_1\subseteq \partial U_1,\Sigma_2\subseteq \partial U_2$, to avoid corners case we suppose $\partial\Sigma_1=\emptyset=\partial \Sigma_2$. This includes an isometry of $\Sigma_1$ with $\Sigma_2$ together with the compatibility of  normal derivatives of the metric. We also suppose that the principal bundle ${\cal P}$ over $U$ is induced by the corresponding principal bundle  ${\cal P}_1,{\cal P}_2$ over $U_1,U_2$. From~the projection map $p:U_1\times U_2\rightarrow U$ we fix base points $\eta\in{\cal A}_{U}$ obtained by gluing $p^*\eta_i\in{\cal A}_{U_i},i=1,2$. 

Suppose that $V_i\in {\frak F}_{U_i},i=1,2$ satisfy the continuity gluing condition along $\Sigma_i$
\begin{equation}\label{eqn:gluing}
	{\sf j}_{\Sigma} \left(V_1\vert_{\Sigma_1}\right)
	=
	{\sf j}_{\Sigma} \left(V_2\vert_{\Sigma_2}\right)
\end{equation}
and denote those couples $(V_1,V_2)$ satisfying (\ref{eqn:gluing}) as ${\frak F}_{U_1}\#_\Sigma{\frak F}_{U_2}$, where $\Sigma=p(\Sigma_i)\subseteq U$. It is a Lie subalgebra of ${\frak F}_{U_1}\oplus {\frak F}_{U_2}$. The continuity gluing condition (\ref{eqn:gluing}) is trivially satisfied for the gauge Lie algebras so that ${\frak G}_{U_1}\#_\Sigma{\frak G}_{U_2}={\frak G}_{U_1}\oplus {\frak G}_{U_2}$, hence there is a well defined Lie algebra
\[
{\frak F}_{U_1}\#_\Sigma{\frak F}_{U_2}/
{\frak G}_{U_1}\#_\Sigma{\frak G}_{U_2}
\subseteq
{\frak F}_{U_1}\oplus{\frak F}_{U_2}/{\frak G}_{U_1}\#_\Sigma{\frak G}_{U_2}
\]

Let $\hat{\frak G}^{\Sigma}_{U_1}\subseteq \hat{\frak G}_{U_1}$ denote those gauge variations whose jet vanish along the boundary components of $\partial U_1$ except for $\Sigma_1$. Similarly define  $\hat{\frak G}_{U_2}^{\Sigma}$. If~we define
\[
	{\frak G}_\Sigma =
	\hat{\frak G}^{\Sigma}_{U_1}\#_\Sigma \hat{\frak G}^{\Sigma}_{U_2}
	/
	\left({\frak G}_{U_1}\#_\Sigma{\frak G}_{U_2}\right)
\]
then by an Isomorphism Theorem for Lie algebras,
\[
			{\frak F}_{U_1}\#_\Sigma{\frak F}_{U_2}/
			\hat{\frak G}^{\Sigma}_{U_1}\#_\Sigma \hat{\frak G}^{\Sigma}_{U_2}
			\simeq 
			\left({\frak F}_{U_1}\#_\Sigma{\frak F}_{U_2}/
			{\frak G}_{U_1}\#_\Sigma{\frak G}_{U_2}\right)
			/
			{\frak G}_\Sigma.
\]

There is a commutative diagram of linear maps as follows. Recall the gluing procedure for abelian YM, see~\cite{DM-O}. The~doted arrow is a Lie algebra morphism.
\[\xymatrix{
			\left({\frak F}_{U_1}\#_\Sigma{\frak F}_{U_2}/
			{\frak G}_{U_1}\#_\Sigma{\frak G}_{U_2}\right)
			/
			{\frak G}_\Sigma
		\ar[ddd]
	&&
	\\
	&
	{\frak F}_{U}/{\frak G}_{U}
			\ar[d]^{{\frak e}_{\eta}}	\ar[r]		\ar@{-->}[ul]
	&
		{\frak F}_{U_1}\#_\Sigma{\frak F}_{U_2}/
		{\frak G}_{U_1}\#_\Sigma{\frak G}_{U_2}
			\ar[d]^{{\frak e}_{\eta_1}\oplus{\frak e}_{\eta_2} }
			\ar[llu]
	\\
	&
	{\cal L}_{U}
		\ar[r]		\ar[dl]^{r_{U,\partial U}}
	&
	{\cal L}_{U_1}\oplus {\cal L}_{U_2}
		\ar[d]^{r_{U_1,\partial U_1}\oplus r_{U_2,\partial U_2}}
	\\
	{\cal L}_{\tilde{U}}
	&&
	{\cal L}_{\tilde{U}_1}\oplus 	{\cal L}_{\tilde{U}_2}
		\ar[ll]	
}\]

From the Lagrangian embedding of ${\cal L}_{\tilde{U}}$ with respect to the symplectic structure, $\omega_{\partial U,L},$ it follows that the Dirichlet conditions along $\Sigma_1$ and $\Sigma_2$ completely determine the Neumann conditions in $U_1$ and $U_2,$ respectively. Here we consider an axial gauge fixing for solutions in $\partial U$ satisfying also the Lorentz gauge fixing condition in $\partial U$, see Appendix in~\cite{DM-O}. This means that the continuous gluing condition (\ref{eqn:gluing}) will suffice to reconstruct modulo gauge the first variation $V_\gf=V_{\gf_1}\#V_{\gf_2}$ for $V_{\gf_i}\in {\frak F}_{U_i},i=1,2$ disregarding the normal derivatives along $\Sigma$. This proves the following~assertion

\begin{theorem}[Gluing of symmetries modulo gauge]
There is an isomorpmhism of Lie algebras
\[
	\left({\frak F}_{U_1}\#_\Sigma{\frak F}_{U_2}/
			{\frak G}_{U_1}\#_\Sigma{\frak G}_{U_2}\right)
			/
			{\frak G}_\Sigma
			\simeq 
	{\frak F}_U/{\frak G}_U.
\]

\end{theorem}



\section{Outlook: Further~Problems}\label{sec:outlook}


We just remark that in further directions of research. In~the first place, it is highly desirable to see whether or not $f^{VV'}_\Sigma$ observables can be defined for non abelian (non-linear) YM equations and if it will suffice to separate solutions just as in Theorem~\ref{tma:darboux}. Extension of the variationa bicomplex ttreatment need to be extended to non-local first variations to get enough observables to separate solutions. The~existence of a Jacobi bracket needs also to be verified in this case. Gluing properties for observables need also to be developed and explained in detail. Namely the continuous gluing \mbox{of currents}
$
	{\sf HOC}_{U_1}\#_\Sigma {\sf HOC}_{U_2}
$
in relation to ${\sf HOC}_{U},$ as well as the gluing ${f}^{V_1}_{\Sigma'}\#_\Sigma{f}_{\Sigma''}^{V_2}$ for hypersurfaces $\Sigma'\subseteq U_1,\Sigma'\subseteq U''$ intersecting transversally the gluing boundary component $\Sigma_i,i=1,2$. Finally, considerations of Riemannian manifolds with corners may introduce further difficulties in the results we have established for the smooth boundary~case.

\vspace{6pt} 


\appendix
\section{Variational Bicomplex~Formalism}\label{VarBiC} 


For the convenience of the reader, in~this section we fix notation by recalling basic definitions of the variational formalism for variational PDEs taken from~\cite{Anderson, Olver, Deligne-Freed, Vinogradov, Krupka, Zuckerman}.

Let $M$ be an $n$-dimensional manifold, and~let $\pi:Y\rightarrow M$ be a fiber bundle with $m$-dimensional fiber $\mathcal{F}$. Denote its sections or {\em histories} as $\Gamma(Y\mid_U)$ where $U\subseteq M$ is a compact domain with piecewise smooth~boundary.

The $k$-jet bundle $\pi_{k,0}:{\sf J}^kY\rightarrow Y$, $k=1,2,\dots$. On~$\pi^{-1}(U)\subset Y$ take the local coordinates
\[
	\left(x;u^{(k)}\right):=
		(x_1,\dots,x_i,\dots,x_n;u^1,\dots,u^a,\dots,u^m;\dots,u^{a}_I,\dots)
			\in {\sf J}^k\pi^{-1}(U)
\]
where $ i=1,\dots, n;\, a=1,\dots,m;$ and $I=(i_1,\dots,i_n)$ denotes a {\em multiindex} of degree $|I|:=i_1+\dots+i_n=0,1,\dots,k$, $i_j\geq 0,i_j\in\mathbb{Z}$. For~$I=\emptyset $, we define $u_\emptyset^a=u^a$. We denote the projection of the $(k+1)$-jet onto the $k$-jet as $\pi_{k+1,k}:{\sf J}^{k+1}Y\rightarrow {\sf J}^kY$. For~a section $\phi:M\rightarrow Y$, we denote its $k$-jet as ${\sf j}^k\phi:M\rightarrow {\sf J}^kY$, where
\[
	{\sf j}^k\phi(x)=
		\left(\phi^1(x),\dots,\phi^m(x);\dots,
		\frac{\partial^{|I|} \phi^{a}}{\partial x_1^{i_1}\dots\partial x_n^{i_n}},\dots\right)
\]





Denote the space of $p$-forms on ${\sf J}^kY$ as ${\mathsf{\Omega}}^p({\sf J}^kY)$. For~the decomposition $p=r+s$, denote the space of $r$-horizontal and  $s$-vertical forms on ${\sf J}^kY$ as ${\mathsf{\Omega}}^{r,s}({\sf J}^kY)$, have as basis the $(r+s)$-forms $\vartheta^{a_1}_{I_1}\wedge\dots\wedge\vartheta^{a_s}_{I_s}\wedge dx_{j_1}\wedge\dots\wedge dx_{j_r}$, where

The Cartan distribution on ${\sf J}^kY$ is generated by the basis of contact $1$-forms (\ref{equation:vertical-basis})
\begin{equation}\label{equation:vertical-basis}
	\vartheta^a_I:=
		du^a_I-\sum_{j=1}^nu^a_{(I,j)}dx_j
			\in{\mathsf{\Omega}}^1\left(J^{|I|+1}Y\right),
				\, 
			|I|\leq k-1,a=1,\dots,m.
\end{equation}

The vertical differential $\mathsf{d_v}$ for $F\in{\mathsf{\Omega}}^0({\sf J}^kY)$ defined as
\[
	\mathsf{d_v}:{\mathsf{\Omega}}^{r,s}\left({\sf J}^kY\right)
		\rightarrow
		{\mathsf{\Omega}}^{r,s+1}\left({\sf J}^kY\right),
	\qquad
	\mathsf{d_v}F
	:=
		\sum_{0\leq|I|\leq k} \sum_{a=1}^m
		\frac{\partial F}{\partial u^a_I}\vartheta^a_I,
\]
then we are forced to consider the horizontal differential with range in the $(p+1)$-forms in ${\sf J}^{k+1}Y$
\[
	\mathsf{d_h}:{\mathsf{\Omega}}^{r,s}\left({\sf J}^kY\right)
		\rightarrow
		{\mathsf{\Omega}}^{r+1,s}\left({\sf J}^{k+1}Y\right),
	\qquad
	\mathsf{d_h}F:=\sum_{i=1}^n\frac{\mathsf{d}}{\mathsf{d}x_i}^{(k+1)}(F)\,dx_i.
\]
where
\[
	\left(\frac{\mathsf{d}}{\mathsf{d}x_i}\right)^{(k)}
		:=
		\frac{\partial}{\partial x_i}
		+
		\sum_{0\leq |J|\leq k-1}\sum_{a=1}^m
			u^a_{(J,i)}\frac{\partial}{\partial u^a_J}.
\]

The injective limit ${\sf \Omega}^{r,s}({\sf J}Y):=\varinjlim_{\pi_{k+1,k}^*} {\sf \Omega}^{r,s}({\sf J}^kY)$, models the $p$ forms in the infinite jet space ${\sf J}Y=\varprojlim_{\pi_{k+1,k}}{\sf J}^kY$. We have the identities
\[
	\mathsf{d_v}^2=0,\qquad\mathsf{d_h}^2=0,\qquad
	\mathsf{d_v}\mathsf{d_h}+\mathsf{d_h}\mathsf{d_v}=0.
\]

Hence, the~following diagram commutes
\[\xymatrix{
\vdots&\vdots&\vdots&
\\
{\mathsf{\Omega}}^{n-2,2}({\sf J}Y)		\ar[r]^{\mathsf{d_h}}		\ar[u]^{\mathsf{d_v}}
&
{\mathsf{\Omega}}^{n-1,2}({\sf J}Y)		\ar[r]^{\mathsf{d_h}}			\ar[u]^{\mathsf{d_v}}
&
{\mathsf{\Omega}}^{n,2}({\sf J}Y)		\ar[u]^{{\sf d_v}}	
\\
{\mathsf{\Omega}}^{n-2,1}({\sf J}Y)		\ar[r]^{\mathsf{d_h}}		\ar[u]^{\mathsf{d_v}}
&
{\mathsf{\Omega}}^{n-1,1}({\sf J}Y)		\ar[r]^{\mathsf{d_h}}	
\ar[u]^{\mathsf{d_v}}
&
{\mathsf{\Omega}}^{n,1}({\sf J}Y)	\ar[u]^{{\sf d_v}}
\\
{\mathsf{\Omega}}^{n-2,0}({\sf J}Y)		\ar[r]^{\mathsf{d_h}}		\ar[u]^{\mathsf{d_v}}
&
{\mathsf{\Omega}}^{n-1,0}({\sf J}Y)		\ar[r]^{\mathsf{d_h}}			\ar[u]^{\mathsf{d_v}}
&
{\mathsf{\Omega}}^{n,0}({\sf J}Y)	\ar[u]^{{\sf d_v}}
\\
{\mathsf{\Omega}}^{n-2}(M)		\ar[r]^{d_{dR}}		\ar@{^{(}->}[u]^{\pi^*}
&
{\mathsf{\Omega}}^{n-1}(M)		\ar@{^{(}->}[u]^{\pi^*}
&
H^{n-1}_{dR}(M)		\ar@{^{(}->}[u]^{\pi^*}
}\]

Derivations in the algebra of smooth functions on ${\sf \Omega}^0({\sf J}Y)$,
\[
		{V} =
			\sum_{i=1}^n
				a_i\left(x;u\right)\frac{\partial}{\partial x_i}
			+
			\sum_{a=1}^m
				V^{a}\left(x;u\right)\frac{\partial}{\partial u^a}
\]
 are in correspondence with sections $V\in \Gamma(\pi_{\infty,0}^*(\pi^{\sf v}))$, where $\pi_{\infty,0}^*(\pi^{\sf v})$ is the pullback under $\pi_{\infty,0}:{\sf J}Y\rightarrow Y$ of the vertical (vector) bundle, $\pi^{\sf v}: Y^{\sf v}\rightarrow Y$, whose fiber at each $(x,u)\in Y$ is consists of the vertical fibers $ Y_{(x,u)}^{\sf v}= T_{(x,u)}\pi^{-1}(x)$. In~fact, its {\em prolongations}
\begin{equation}\label{equation:prolongation}
		{\sf j}^kV=
			\sum_{i=1}^na_i\left(x;u\right)	\left(\frac{\mathsf{d}}{\mathsf{d}x_i}\right)^{(k)}+
			\sum_{0\leq |I| \leq k-1}
			\sum_{a=1}^mb^{a}_I\left(x;u^{(k)}\right)\frac{\partial}{\partial u^a_I}		
\end{equation}
where $b_I^a:= D_I^{(k)}V^a+\sum_{j=1}^nu^a_{jI}a_j$, $|I|\leq k-1$, act as infinitesimal symmetries of the Cartan distribution in ${\sf J}Y$ in the sense that
\begin{equation}\label{eqn:contactomorphism}
	\mathscr{L}_{{\sf j}V}\vartheta^a_I
	=
	0.
\end{equation}

Here the horizontal derivative operator $D_I^{(k)}$ equals $
D_{i_1}^{(k)}\circ D_{i_2}^{(k)}\circ \dots \circ D_{i_n}^{(k)},$ with
\[
	\,
		D_{i_s}^{(k)}
		=
		\left(\frac{\mathsf{d}}{\mathsf{d}x_s}\right)^{(k)}
                \circ\dots\circ
 	        \left(\frac{\mathsf{d}}{\mathsf{d}x_s}\right)^{(k)}
		=
		\left(
                \left(\frac{\mathsf{d}}{\mathsf{d}x_s}\right)^{(k)}
                \right)^{\circ i_s},
\]


We will assume that $V\in  \Gamma(\pi_{\infty,0}^*(\pi^{\sf v}))$ has no horizontal component. Hence $a_i=0$ and
\begin{equation}\label{equation:prolongation-phi}
		{\sf j}V
		=
                \sum_{a=1}^m\left(
                V^{a}\frac{\partial}{\partial u^a}
                +
                \sum_{i=1}^n
                \left(\frac{\mathsf{d}}{\mathsf{d}x_i}\right)^{(k)}
                (V^{a})\frac{\partial}{\partial u^a_i}
                +
		\sum_{2\leq |I| \leq k}
		D_I^{(k)}(V^{a})\frac{\partial}{\partial u^a_I}\right).	
\end{equation}

We call this space the space of {\em evolutionary vector fields},
\begin{equation}\label{eq:Phi}
	\mathfrak{Ev}({\sf J}Y)=
	\left\{
		V\in\Gamma\left(\pi_{\infty,0}^*(\pi^{\sf v} )\right)
		\,:\,
		\frac{\partial V^a}{\partial u^a_i}
	=
	\frac{\partial V^b}{\partial u^b_i},
	\,
	\frac{\partial V^b}{\partial u^a_i}=0,
	\,
	a\neq b
	\right\}
\end{equation}
where the functions $V^a$ are {\em local} in the sense that they depend on a finite number of derivatives of $u$.

For a first-order Lagrangian variational problem in a region $U\subseteq M$, the~{\em space of first variations of histories}, $v=\delta\phi$, for~a fixed $\phi\in\Gamma(Y\vert_U)$, can be modeled as
\[
		V
			\in
		\left\{
			V\in\Gamma\left(\pi_{1,0}^*(\pi )\vert_{\cal U}\right)
			\,:\,
			V\in \mathfrak{Ev}\left( {\cal U}\subseteq {\sf J}Y \right)
		\right\}
\]
 where $\mathcal{U}\subseteq {\sf J}Y,$ is a neighborhood of the graph ${{\sf j}\phi(U)}$. 






Let $\nu=dx_1\wedge\dots\wedge dx_n\in{\mathsf{\Omega}}^n(M)$ be a fixed volume $n$-form on $M$, and~consider the Lagrangian density, $L={\rm L}\cdot \nu\in {\mathsf{\Omega}}^{n,0}({\sf J}^1Y)$, with~Lagrangian
\[
	{\rm L}={\rm L}(x_i,u^a,u^a_i)\in{\mathsf{\Omega}}^0({\sf J}^1Y).
\]

Consider the action functional on $U\subset M$, 
\[
	{S}_{U}(\phi)=\int_U({\sf j}^1\phi)^*{L}
\]
if we take the vertical derivative
\begin{equation}\label{eq:0}
	\pi_{2,1}^*(\mathsf{d_v}{L})
	=
	\mathsf{d_h}\Theta_L + E({L})
\end{equation}
\[
	\Theta_L=
		-
		\sum_{i=1}^n\sum_{a=1}^m
			\frac{\partial {\rm L}}{\partial u_i^a}\vartheta^a\wedge\nu_i
				\in {\mathsf{\Omega}}^{n-1,1}\left({\sf J}^1Y\right),\qquad
					dx_i\wedge \nu_i=\nu
\]
and $
	E(L)\left(x;u^{(2)}\right)=
		\sum_{a=1}^mE_a({\rm L})\cdot \vartheta^a\wedge\nu
			\in {\mathsf{\Omega}}^{n,1}({\sf J}^2Y),$
then the Euler-Lagrange equations are
\[
	E_a({\rm L})	=
		\frac{\partial {\rm L}}{\partial u^a}-\sum_{i=1}^n\frac{\mathsf{d}}{\mathsf{d}x_i}^{(2)}\left(\frac{\partial {\rm L} }{\partial u^a_i}\right)	=
			0,
			\qquad a=1,\dots,m.
\]

Recall that $({\mathsf{d}}/{\mathsf{d}x_i})^{(2)}=\partial /\partial x_i+\sum_b \left(u^b_i\partial/\partial u^b+\sum_ju^b_{ij}\partial/\partial u_j^b\right)$. Another way of obtaining the Euler-Lagrange equations is by $E(L)={\sf I}({\sf d_v}L)$ where we use the integration by parts operator $
	{\sf I}:{\sf \Omega}^{n,s}({\sf J}^1 Y)\rightarrow {\sf \Omega}^{n,s}({\sf J}^2 Y)
$ for $s>0$, satisfying ${\sf I}\circ {\sf d_h}=0,$ ${\sf I}\circ{\sf I}={\sf I}$ and that $\alpha-{\sf I}\alpha$ is ${\sf d_h}$-exact. In~coordinates ${\sf I}$ it is given by 
\[
	{\sf I}(\alpha)
	=
	\frac{1}{s}\vartheta^a\wedge 
	\sum_{a=1}^m
	\left[
		\iota_{\frac{\partial}{\partial u^a}}\alpha
		-
		\sum_{j=1}^n
		\frac{\sf d}{{\sf d} x_j}\left(\iota_{\frac{\partial}{\partial u^a_j}}\alpha\right)
	\right]
.\]

Meanwhile, the~{\em locus of the Euler-Lagrange} PDEs
\[
	\mathcal{E}_L:=
		{\sf j}\left\{
			\left(x;u^{(2)}\right)\in {\sf J}^2Y		\,:\,
			E(L)\left(x;u^{(2)}\right)=0
		\right\}
			\subseteq {\sf J}Y.
\]

The {\em space of solutions} of the Euler-Lagrange equations
\[
	{\cal A}_U:=
		\left\{\phi\in\Gamma(Y)\,:
			\, {\sf j}\phi(M)\subseteq  \mathcal{E}_L
	 \right\}.
\]

On the other hand, if~we define the form $\Omega_L=-\mathsf{d_v}\Theta_L\in {\mathsf{\Omega}}^{n-1,2}\left({\sf J}^1Y\right)$, or~\[
		\Omega_L=
		\sum_{i=1}^n
		\left(
			\sum_{b,a=1}^m
			\left(
					\frac{\partial^2 {\rm L}}{\partial u^b\partial u^a}
					\vartheta^b\wedge\vartheta^a
			+
				\sum_{j=1}^n
					\frac{\partial^2{\rm L}}{\partial u^b_j\partial u^a_i}
					\vartheta_j^b\wedge\vartheta^a
			\right)
		\right)\wedge\nu_i.
\]

For a first variation $\delta \phi$ modeled by $V\in\mathfrak{Ev}({\sf J}Y\vert_U)$, let us consider the Cartan formula for \mbox{vertical derivation}
\begin{equation}\label{eqn:Cartan-vertical}
	\mathscr{L}_{V}b
		=
	\mathsf{d_v}\left(\iota_{{\sf j}V} b\right)+\iota_{{\sf j}V} \mathsf{d_v} b
\end{equation}
see~\cite{Anderson} Proposition {1.16}. Then
\begin{equation}\label{eqn:Cartan}	
	\pi_{2,1}^*\left(\mathscr{L}_{{\sf j}V}b\right)
	=
	{d}\left(
		\pi_{2,1}^*\left(
			\iota_{{\sf j}^1 V}b
		\right)
	\right)
	  	+
	\iota_{{\sf j}^1V}({d}(\pi_{2,1}^* b))
\end{equation}

Therefore,
\[
	\mathsf{d_h}\left(\iota_{{\sf j}^1V} b\right)
	=
	(d\circ\pi_{2,1}^*-\pi_{2,1}^*\circ\mathsf{d_v})	\left(\iota_{{\sf j}^1V} b\right)
	=
\]
\[
	-	\left(
		\iota_{{\sf j}^2V}
			(d\circ \pi_{2,1}^*-\pi_{2,1}^*\circ \mathsf{d_v})
		\right)
			b
	=
	-
	\iota_{{\sf j}^2V}	\left(\mathsf{d_h}b\right)
\]	
or
\begin{equation}\label{eqn:d_h-iota}
	\mathsf{d_h}\left(\iota_{{\sf j}^1V}(\cdot)\right)
	=
	-
	\iota_{{\sf j}^2V}	\left(\mathsf{d_h}(\cdot)\right).
\end{equation}

In particular $
	\mathsf{d_h}\left(\iota_{{\sf j}^1V}\Theta_L\right)
	=
	-
	\iota_{{\sf j}^2V}	\left(\mathsf{d_h}\Theta_L\right)$.
Hence the variation for the action is
\[
		\left.\frac{d}{d\geps}\right\vert_{\geps=0}
			{S}_{U}(\phi_\geps)
	=
		\left.\frac{d}{d\geps}\right\vert_{\geps=0}
			\int_U({\sf j}^1\phi_\geps)^*{{L}}
	=
\]
\[
		 \int_U{\sf j}^1\phi^*\left(\mathscr{L}_{{\sf j}^1V}{L}\right)
	=
		 \int_U{\sf j}^1\phi^*
		 	\left(
		 		\iota_{{\sf j}^1V} \mathsf{d_v}L
			\right)
	=
\]
\[
		 \int_U{\sf j}^2\phi^*
		 	\left(
		 		\iota_{{\sf j}^2V} \pi_{2,1}^*(\mathsf{d_v}{L})
			\right)
	=
\]
\[
		 \int_U{\sf j}^2\phi^*\left(\iota_{j^2X}\mathsf{d_h}\Theta_L
		 +
				\iota_{{\sf j}^2V} E(L)
		 \right)
	=
\]
\[
		-
		\int_{ U}{\sf j}^2\phi^*(\mathsf{d_h}+\pi_{2,1}^*\circ \mathsf{d_v})\left(\iota_{{\sf j}^1V}\Theta_L\right) 
			+ 
		 \int_U{\sf j}^1\phi^*\left(\iota_{{\sf j}^2V}E(L)\right)
\]
\[
		+
		 \int_U{\sf j}^1\phi^*\left(\mathsf{d_v}(\iota_{{\sf j}^1V}\Theta_L)\right)
\]
\[
	=
	-
		\int_{ U}{\sf j}^2\phi^*d\left(\iota_{{\sf j}^1V} \Theta_L\right) 
			+ 
		 \int_U{\sf j}^2\phi^*\left(\iota_{{\sf j}^2V} E(L)\right)
\]
\[
=	-\int_{ \partial U}{\sf j}^1\phi^*\left(\iota_{{\sf j}^1V}\Theta_L\right) 
			+ 
		 \int_U{\sf j}^2\phi^*\left(\iota_{{\sf j}^2V} E(L)\right).
\]

\begin{proposition}\label{pro:partial-EL}
Let $\hat{\Theta}_L = - \Theta_L + L$ be for the Poincar\`e-Cartan form, which is also the principal Lepage equivalent of $\Theta_L$, and~let $\phi$ be a section. The~following assertions are~equivalent
\begin{enumerate}
\item $\phi\in {\cal A}_U$.

\item For every vertical vector field, $V\in\mathfrak{Ev}({\sf J}Y\vert_U)$,  the~$n$-form $({\sf j}\phi)^*\iota_{{\sf j}V}\hat{\Omega}_L$ in $U\subseteq M$,  vanishes.

\item The Euler-Lagrange equations hold for every $x\in U$
\[
	\frac{\partial {\rm L}}{\partial u^a}({\sf j}\phi(x))
		-
	\sum_{i=1}^n 
	 	\frac{\partial^2 {\rm L}}{{\partial x_i}{\partial u^a_i}}({\sf j}\phi(x))
	=0
		,	\qquad
	\forall
	a=1,\dots, m.
\]
\end{enumerate}
\end{proposition}
\vspace{-6pt}

Notice that in the Euler-Lagrange equations $E_a({\rm L})=0$ arising from ${\sf j}^2\phi^*E(L)=0$, the~total horizontal derivations ${\sf d}/{\sf d}x_i$ are involved. Meanwhile, the Euler-Lagrange equations mentioned in Proposition \ref{pro:partial-EL} deal with partial horizontal derivations, $\partial/\partial x_i$, see~\cite{Forger, Sternberg-G}.

In general for an $(n-1)$-dimensional manifold $\Sigma\subseteq U$, we can define the $1$-form
\[
	(\theta_{L\Sigma})_\phi(v)
	:=
	\int_{\Sigma}{\sf j}\phi^*\left(\iota_{{\sf j}V}\Theta_L\right),
		\qquad
		\forall	\phi, 	v=\delta{\phi},
\]
where the variation $v=\delta\phi$ corresponds to $V\in \mathfrak{Ev}({\sf J}Y\vert_U)$. For~$\phi\in{\cal A}_U$,
\[
	\left(\mathrm{d}S_{U}\vert_{{\cal A}_U}\right)_\phi(v)
	=
	-\left(\left.\theta_{L,\partial U}\right\vert_{{\cal A}_U}\right)_\phi(v)
\]

Define the {\em presymplectic structure}  $\omega_{L\Sigma}:=-\mathrm{d}\theta_{L\Sigma}$, $\forall \phi,$ so that $
	\forall v=\delta\phi,v'=(\delta\phi)'$ we have
\[
	(\omega_{L\Sigma})_\phi(v,v')
	=
	\int_{\Sigma}{\sf j}\phi^*\left(\iota_{{\sf j}V'}\iota_{{\sf j}V}\Omega_L\right).
\]

From $\mathsf{d_v}\Omega_L=0$ it follows that $\mathrm{d}\omega_{L\Sigma}=0$.


\section*{Acknowledgements}
{This research was funded by CONACYT-MEXICO, grant number~58132.}

{The author thanks J. A. Zapata since most of the results of this article arise as a particular application of the results sketched in a joint work~\cite{DM-Z} see also~\cite{Zapata(2019)}. Along the review of~\cite{DM-Z} many clarifications came from correspondence with Luca Vitagliano and notes from an anonymous referee who provided many clarifications for the difficulties arising in non-linear field~theories.} {The author declare no conflict of~interest.} 


The following abbreviations are used in this manuscript:\\

\noindent 
\begin{tabular}{@{}ll}
PDE & Partial differential equation\\
YM & Yang-Mills\\
BVP & Boundary value problem\\
\end{tabular}}



\label{last}
\end{document}